\documentclass[a4paper,onecolumn,accepted=2021-01-09]{quantumarticle}

\pdfoutput = 1

\usepackage[usenames,dvipsnames]{xcolor}
\usepackage{amsfonts}
\usepackage{amsmath,amsthm,amssymb,dsfont}
\usepackage{enumerate}
\usepackage[english]{babel}
\usepackage{graphicx}	
\usepackage{url}
\usepackage{todonotes}
\usepackage{bbm}
\usepackage{booktabs}

\usepackage[numbers,sort&compress]{natbib}

\usepackage{tikz}
\usetikzlibrary{chains}
\usetikzlibrary{fit}
\usepackage{pgflibraryarrows}		
\usepackage{pgflibrarysnakes}		

\usepackage{makecell}

\usepackage{epsfig}
\usetikzlibrary{shapes.symbols,patterns} 
\usepackage{pgfplots}
\pgfplotsset{compat=newest}

\usepackage{mathrsfs}

\usepackage{hyperref}
\hypersetup{colorlinks=true,citecolor=blue,linkcolor=blue,filecolor=blue,urlcolor=blue,breaklinks=true}

\usepackage{nicefrac}
\usepackage{mathtools}

\theoremstyle{plain}
\newtheorem{theorem}{Theorem}[section]
\newtheorem{lemma}[theorem]{Lemma}

\newtheorem{corollary}[theorem]{Corollary}
\newtheorem{proposition}[theorem]{Proposition}

\newtheorem*{propositionno}{Proposition}

\theoremstyle{definition}

\newtheorem{remark}[theorem]{Remark}

\newcommand*{\cA}{\mathcal{A}}
\newcommand*{\cB}{\mathcal{B}}
\newcommand*{\cC}{\mathcal{C}}
\newcommand*{\cD}{\mathcal{D}}

\newcommand*{\cF}{\mathcal{F}}

\newcommand*{\cH}{\mathcal{H}}

\newcommand*{\cI}{\mathcal{I}}

\newcommand*{\cN}{\mathcal{N}}
\newcommand*{\cM}{\mathcal{M}}
\newcommand*{\cP}{\mathcal{P}}

\newcommand*{\cT}{\mathcal{T}}

\newcommand*{\cV}{\mathcal{V}}

\newcommand*{\RR}{\mathbb{R}}

\newcommand*{\CC}{\mathbb{C}}

\newcommand*{\NN}{\mathbb{N}}

\newcommand*{\eps}{\varepsilon}

\newcommand*{\id}{I}

\newcommand*{\spec}{\mathrm{spec}}
\newcommand*{\supp}{\mathrm{supp}}
\newcommand*{\tr}{\mathrm{tr}\,}
\newcommand*{\ket}[1]{| #1 \rangle}
\newcommand*{\bra}[1]{\langle #1 |}

\newcommand{\ketbra}[2]{|#1\rangle\!\langle #2|}
\newcommand{\proj}[1]{|#1\rangle\!\langle #1|}

\newcommand*{\Pos}{\mathscr{P}}
\newcommand*{\Lin}{\mathscr{L}}

\newcommand*{\D}{\mathscr{D}}

\newcommand*{\CP}{\mathrm{CP}}

\newcommand{\newD}{D^{\#}}
\newcommand{\newQ}{Q^{\#}}

\newcommand{\reg}{\mathrm{reg}}
\newcommand{\amor}{\mathrm{a}}

\newcommand{\sharpf}{\#_f}

\newenvironment{sm}{\left[\begin{smallmatrix}}{\end{smallmatrix}\right]}

\DeclareMathOperator*{\argmin}{argmin}

\usepackage[normalem]{ulem}
\renewcommand\sout[1]{\ignorespaces}
\newcommand{\change}[1]{{\color{magenta}#1}}
\newcommand{\changecomment}[1]{\marginpar{\tiny \color{magenta} #1}}

\renewcommand{\sout}[1]{}
\renewcommand{\change}[1]{#1}
\renewcommand{\changecomment}[1]{}

\begin{document}

\title{Defining quantum divergences via convex optimization}

\author[1]{Hamza Fawzi}
\author[2]{Omar Fawzi}

\affil[1]{\small{DAMTP, University of Cambridge, United Kingdom}}
\affil[2]{\small{Univ Lyon, ENS Lyon, UCBL, CNRS, Inria, LIP, F-69342, Lyon Cedex 07, France}}

\maketitle


\begin{abstract}
We introduce a new quantum R\'enyi divergence $\newD_{\alpha}$ for $\alpha \in (1,\infty)$ defined in terms of a convex optimization program. This divergence has several desirable computational and operational properties such as an efficient semidefinite programming representation for states and channels, and a chain rule property. An important property of this new divergence is that its regularization is equal to the sandwiched (also known as the minimal) quantum R\'enyi divergence. This allows us to prove several results. First, we use it to get a converging hierarchy of upper bounds on the regularized sandwiched $\alpha$-R\'enyi divergence between quantum channels for $\alpha > 1$. Second it allows us to prove a chain rule property for the sandwiched $\alpha$-R\'enyi divergence for $\alpha > 1$ which we use to characterize the strong converse exponent for channel discrimination. Finally it allows us to get improved bounds on quantum channel capacities.

\end{abstract}


\section{Introduction}

Given nonnegative vectors $P, Q \in \RR^\Sigma$, the $\alpha$-R\'enyi divergence is defined as 
\begin{align}
D_{\alpha}(P \| Q) &=
\left\{\begin{array}{ll}
        \frac{1}{\alpha - 1} \log \sum_{x \in \Sigma} P(x)^{\alpha} Q(x)^{1-\alpha}  & \text{if } P \ll Q \\
        \infty & \text{else}.
        \end{array} \right.
        \label{eq:def_classical_div}
\end{align}
for $\alpha \in (1, \infty)$. Here $P \ll Q$ means that for any $x \in \Sigma$, $Q(x) = 0$ implies that $P(x) = 0$ and the $\log$ is taken to be base $2$. This definition has found many applications in information theory and beyond, we refer to the survey paper~\cite{EH14} for more general definitions and properties of this quantity. To generalize this notion to quantum states $\rho$ and $\sigma$ which are now positive semidefinite operators on $\CC^{\Sigma}$ the interpretation of the multiplication appearing in the definition matters and multiple definitions exist. Such definitions are systematically studied in~\cite{Tom15book}. We mention two important examples for our work. For positive semidefinite operators $\rho$ and $\sigma$ on $\CC^{\Sigma}$, provided $\rho \ll \sigma$ (i.e., the support of $\rho$ is contained in the support of $\sigma$), the geometric and sandwiched divergences 
\begin{align}
\widehat{D}_{\alpha}(\rho \| \sigma) &= \frac{1}{\alpha - 1} \log \tr\left( \sigma^{\frac{1}{2}} (\sigma^{-\frac{1}{2}} \rho \sigma^{-\frac{1}{2}})^{\alpha} \sigma^{\frac{1}{2}}  \right) \label{eq:def_geom_div} \\
\widetilde{D}_{\alpha}(\rho \| \sigma) &= \frac{1}{\alpha - 1} \log \tr\left(  \left(\sigma^{\frac{1-\alpha}{2\alpha}} \rho \sigma^{\frac{1-\alpha}{2\alpha}} \right)^{\alpha} \right) \label{eq:def_minim_div} \ ,
\end{align}
were respectively defined by~\cite{Mat13} and by~\cite{MDSFT13,WWY13}.
The inverses here should be understood as generalized inverses i.e., the inverse on the support. When $\rho \ll \sigma$ is not satisfied, both quantities are set to $\infty$. Whenever $\rho$ and $\sigma$ commute, both definitions agree and as $\rho$ and $\sigma$ can be diagonalized in the same basis this also matches with the classical definition~\eqref{eq:def_classical_div}. For this reason, throughout the paper, if $\rho$ and $\sigma$ commute, then we simply write $D_{\alpha}(\rho \| \sigma)$ for the classical $\alpha$-R\'enyi divergence in definition~\eqref{eq:def_classical_div}.

Another natural definition is the measured R\'enyi divergence, which is obtained by performing the same measurement on $\rho$ and $\sigma$ and then considering the classical R\'enyi divergence after performing the measurement:
\begin{align}
D^{\mathbb{M}}_{\alpha}(\rho \| \sigma) &= \sup_{\{\ket{v_x}\}_{x \in \Sigma}} D_{\alpha}\left( \sum_{x} \proj{v_x} \rho \proj{v_x} \; \Big\| \; \sum_{x} \proj{v_x} \sigma \proj{v_x} \right) \label{eq:def_meas_div} \ ,
\end{align}
where the supremum is chosen over all orthonormal bases $\{\ket{v_x}\}_{x}$ of $\CC^{\Sigma}$. This definition was proposed by~\cite{Don86,HP91} and we refer to~\cite{BFT15} for the equivalence between the different variants.
 
\paragraph{Contributions} In this paper, we put forward another way of defining quantum R\'enyi divergences through a convex optimization program.   Even if such divergences may not have operational interpretations in terms of some information processing task, we demonstrate in this paper that they can nonetheless be useful tools for proofs and for computations.
Given $\alpha \in (1,\infty)$, we define the $\#$ R\'enyi divergence of order $\alpha$ between two positive semidefinite operators $\rho, \sigma$ as
\begin{align}
\newD_{\alpha}( \rho \| \sigma) &:= \frac{1}{\alpha - 1} \log \newQ_{\alpha}(\rho \| \sigma) \ , \notag \\
\newQ_{\alpha}(\rho \| \sigma) &:= \inf_{A \geq 0} \tr(A)  \quad \text{ s.t. } \quad \rho \leq \sigma \#_{\frac{1}{\alpha}} A \ .\label{eq:expr_min_def}
\end{align}
Here $\sigma \#_{\frac{1}{\alpha}} A$ denotes the $\frac{1}{\alpha}$-geometric mean of $\sigma$ and $A$. We recall the definitions and properties of the matrix geometric mean in Section \ref{sec:gm} below. Using the joint concavity of the matrix geometric mean, the optimization program in~\eqref{eq:expr_min_def} is convex and for rational values of $\alpha$ it can be expressed as a semidefinite program~\cite{Sag13,FS17}.
We show the following properties:
 
\begin{itemize}
\item We prove in Section~\ref{sec:states} that $\newD_{\alpha}$ satisfies the data processing inequality and it matches with $D_{\alpha}$ for commuting operators. In addition, it is subadditive under tensor product and when regularized, it is equal to the sandwiched divergence i.e., we have (Proposition \ref{prop:relation_other_div})
\[
\lim_{n\rightarrow\infty} \frac{1}{n} \newD_{\alpha}(\rho^{\otimes n} \| \sigma^{\otimes n}) \;\; = \;\; \widetilde{D}_{\alpha}(\rho\|\sigma).
\]
\item We establish in Section~\ref{sec:channels} that the extension of $\newD_{\alpha}$ to channels has an expression as a convex optimization program (similar to the one for states~\eqref{eq:expr_min_def}) and satisfies subadditivity under tensor product as well as a chain rule property. Furthermore, when regularized, it gives the regularized sandwiched divergence between channels.
\end{itemize}
We then give some applications of $\newD_{\alpha}$ in Section~\ref{sec:applications}.
\begin{itemize}
\item We show that, for $\alpha > 1$, the regularized sandwiched $\alpha$-R\'enyi divergence between quantum channels can be computed to arbitrary precision in finite time (Theorem \ref{thm:comp_reg_sand}).
\item We prove a new chain rule property for $\widetilde{D}_{\alpha}$ for $\alpha > 1$ (Corollary~\ref{cor:chainrulesandwiched}). In turn, the new chain rule property allows us to characterize the strong converse exponent for channel discrimination and show that in this regime adaptive strategies do not offer an advantage over nonadaptive strategies (Section~\ref{sec:channel_discrimination}).
\item We give improved bounds on amortized entanglement measures, which can be used for example to bound the quantum capacity of channels with free two-way classical communication (Section~\ref{sec:amortized_entanglement}). We restrict our focus in this paper to the resource of entanglement, but we expect the same techniques to be applicable in many other resource-theoretic frameworks~\cite{chitambar2019quantum}.
\end{itemize}
 
In addition to these applications, we mention that a close variant of this divergence is introduced in~\cite{BFF20} to bound conditional entropies for quantum correlations.

\begin{remark}[Convention for unnormalized $\rho$]
We note that divergences are most interesting when the first argument is normalized, i.e., $\tr(\rho) = 1$ but it is convenient to keep the definition general. To define it for a general positive semidefinite operator $\rho$, we use a convention which is not standard. We choose for this work to use the exact same expression, e.g.,~\eqref{eq:def_classical_div} for the classical case,~\eqref{eq:def_minim_div} for the sandwiched R\'enyi divergence etc... even if $\rho$ is not normalized. With this convention, $\widetilde{D}_{\alpha}(\rho \| \sigma) = \widetilde{D}_{\alpha}(\frac{\rho}{\tr(\rho)} \| \sigma) + \frac{\alpha}{\alpha-1} \log \tr(\rho)$, which is slightly different from the more standard choice made for example in~\cite{Tom15book} where the correction term for normalization is simply $\log \tr(\rho)$. Note however that the difference between these variants only depends on $\tr(\rho)$ and $\alpha$, and thus the two variants basically have the same properties even when $\rho$ is not normalized. In particular, we will be using the property that the regularized measured divergence \change{is equal to} the sandwiched divergence, a property which clearly holds equally well for both conventions.
\end{remark}

\subsection{Notation}

Let $\cH$ be a finite dimensional Hilbert space and we write $\Lin(\cH)$ for the set of linear operators on $\cH$, $\Pos(\cH)$ for the set of positive semidefinite operators on $\cH$ and $\D(\cH) = \{\rho \in \Pos(\cH) : \tr(\rho)=1\}$. For $A, B \in \Lin(\cH)$, we write $A \geq B$ if $A-B \in \Pos(\cH)$. \change{We let $\|A\|_{\infty} = \max_{\|\ket{\psi}\|=1} \|A \ket{\psi}\|$ be the operator norm of $A$.} Also, for positive semidefinite operators $A$ and $B$, we write $A \ll B$ when $\supp(A) \subseteq \supp(B)$, where $\supp(A)$ denotes the support of $A$. 
 We denote by $\spec(A)$ the spectrum of $A$. For $\rho, \sigma \in \Pos(\cH)$ with $\rho \ll \sigma$, we write $D_{\max}( \rho \| \sigma) = \log \inf\{\lambda \in \RR : \rho \leq \lambda \sigma\}$. When $\cH = X \otimes Y$ for some Hilbert spaces $X$ and $Y$, we often explicitly indicate the systems by writing $A_{XY}$ for $A \in \Pos(X \otimes Y)$. Then $A_{X}$ denotes $\tr_{Y}(A_{XY})$. 

We denote by $\CP(X, Y)$ the set of completely positive maps from $\Lin(X)$ to $\Lin(Y)$. For $\cN \in \CP(X, Y)$, we denote by $J^{\cN}_{XY} \in \Pos(X \otimes Y)$ the corresponding Choi state defined by $J^{\cN}_{XY} = (\cI_{X} \otimes \cN)(\Phi_{XX'})$. 
 Here $X'$ is a copy of the space $X$, $\Phi_{XX'} = \sum_{x,x'} \ketbra{x}{x'}_{X} \otimes \ketbra{x}{x'}_{X'}$, where $\{\ket{x}\}_{x}$ labels a fixed basis of $X$ and $X'$ and $\cI_{X}$ denotes the identity map on $\Lin(X)$.

%

\section{Geometric means and the Kubo-Ando theory}

\label{sec:gm}


In~\cite{kuboando}, Kubo and Ando developed a general theory of operator means from operator monotone functions. The goal of this section is to recall the properties of these means which will be useful for the rest of this paper. This paper will deal with the operator means obtained from the operator monotone functions $f(x) = x^{\beta}$ for $\beta \in [0,1]$ (the so-called $\beta$-matrix geometric mean), however we keep the discussion general as we believe other choices of $f$ can be useful.

Given an operator monotone function $f:[0,\infty) \rightarrow [0,\infty)$ such that $f(1)=1$, the Kubo-Ando mean $\sharpf$ is defined for any pair of positive semidefinite operators $A,B$ and satisfies the following properties \cite{kuboando}, see also \cite[Theorem 37.1 and the following discussion]{simonchapter}:

\begin{itemize}
\item[(i)] Monotonicity: $A \leq C$ and $B \leq D$ implies $A \sharpf B \leq C \sharpf D$ 

\item[(ii)] Transformer inequality: $M(A \sharpf B)M^* \leq (MAM^*) \sharpf (MBM^*)$, with equality if $M$ is invertible 

\item[(iii)] Continuity: if\footnote{By $A_n \downarrow A$ we mean $A_1 \geq A_2 \geq A_3 \geq ...$ and $A_n \rightarrow A$} $A_n \downarrow A$ and $B_n \downarrow B$ then $A_n \sharpf B_n \downarrow A\sharpf B$ 

\item[(iv)] $A \sharpf A = A$

\item[(v)] Joint-concavity:\footnote{Because the map $(A,B) \mapsto A \sharpf B$ is positively homogeneous, joint-concavity is equivalent to the inequality \eqref{eq:concavitygm}} for any $A_i, B_i \geq 0$ we have
\begin{equation}
\label{eq:concavitygm}
\sum_{i} A_i \sharpf B_i \leq \left(\sum_{i} A_i\right) \sharpf \left(\sum_{i} B_i\right)  \ .
\end{equation}

\item[(vi)] For invertible $A$, we have
\begin{equation}
\label{eq:geominv}
A\sharpf B = A^{1/2} f\left(A^{-1/2} B A^{-1/2}\right) A^{1/2}.
\end{equation}
\end{itemize}

Note that properties (ii) and (v) immediately imply that if $\cN$ is a completely positive map, then $\cN(A \sharpf B) \leq \cN(A) \sharpf \cN(B)$.
\change{In fact it is known that the inequality above is true even if we only assume that $\cN$ is a positive map (instead of completely positive), see e.g., \cite[Proposition 3.30]{hiai2017different} or \cite[Lemma 6.3]{Mat13}. 
\begin{proposition}
\label{prop:monotonepositive}
If $\cN$ is a positive map and $A,B \geq 0$ then $\cN(A \sharpf B) \leq \cN(A) \sharpf \cN(B)$.
\end{proposition}
}
\change{
The Kubo-Ando mean has the integral representation \cite{kuboando}
\begin{equation}
\label{eq:intrepsharpf}
A \#_f B = f(0) A + \left( \lim_{x\downarrow0} xf(x^{-1})\right)B + \int_{0}^{\infty} \frac{1+t}{t} (tA) : B d\mu(t)
\end{equation}
for some measure $\mu$ on $(0,\infty)$ depending on $f$, and where for $A',B \geq 0$, $A':B$ denotes the \emph{parallel sum} of $A'$ and $B$ which satisfies \cite[Theorem 9]{andersontrapp}
\begin{equation}
\label{eq:varpar}
\langle x , (A':B) x \rangle = \inf_{\substack{y,z \in \cH\\ y+z = x}} \langle y, A' y \rangle + \langle z, B z \rangle.
\end{equation}
}
%

Some additional properties of the Kubo-Ando mean will be needed in this paper.

\begin{proposition}
\label{prop:p}
For any operator monotone function $f:[0,\infty)\rightarrow [0,\infty)$, the following properties hold:

(vii) If $B \ll A$ then the formula \eqref{eq:geominv} is still valid provided we use the generalized inverse for $A$.

(viii) Direct sum: $(A_1 + A_2) \sharpf (B_1 + B_2) = (A_1 \sharpf B_1) + (A_2 \sharpf B_2)$, for any $A_1,A_2,B_1,B_2 \geq 0$ such that $\supp(A_1 + B_1) \perp \supp(A_2 + B_2)$.

(ix) If $f$ is such that $\lim_{x \downarrow 0} x f(\frac{1}{x}) = 0$ then $A\sharpf B \ll A$ and if $f(0) = 0$ then $A \sharpf B \ll B$, for any $A,B \geq 0$. 

(x) If $B_n \to B$ and $B_n \ll A$ then $A \sharpf B_n \to A\sharpf B$.
\end{proposition}

\begin{proof}

(vii) \change{Let $P$ be the projector onto the orthogonal complement of $\supp(A)$. For any $\eps > 0$, $A+\eps P$ is invertible, thus we have, by \eqref{eq:geominv}
\[
(A+\eps P) \sharpf B = (A+\eps P)^{1/2} f\left((A+\eps P)^{-1/2} B (A+\eps P)^{-1/2}\right) (A+\eps P)^{1/2}.
\]
Since $A$ and $P$ have orthogonal supports we have $(A+\eps P)^{-1/2} = A^{-1/2} + \eps^{-1/2} P$ where $A^{-1/2}$ is the square root of the generalized inverse of $A$. It follows, since $B \ll A$, that $(A+\eps P)^{-1/2} B (A+\eps P)^{-1/2} = A^{-1/2} B A^{-1/2}$. This implies
\[
(A+\eps P)\sharpf B = (A^{1/2} + \eps^{1/2} P) f\left(A^{-1/2} B A^{-1/2}\right) (A^{1/2} + \eps^{1/2} P).
\]
Letting $\eps\to0$, and using the continuity property (iii) we get the desired equality.
}
\sout{After a suitable unitary we can assume that $A$ and $B$ are in blocks \mbox{$A = \begin{sm} A_1 & 0\\ 0 & 0\end{sm}$} and \mbox{$B = \begin{sm} B_1 & 0\\ 0 & 0\end{sm}$} where $A_1$ is invertible. Using the fact that $0 \sharpf 0 = 0$, we have $A\sharpf B = (A_1 \sharpf B_1) \oplus 0$ and $A_1 \sharpf B_1$ is given by the formula \eqref{eq:geominv}. This is exactly what one gets by using the formula directly with $A$ and $B$, and using the generalized inverse.}

(viii) \change{We first assume that $A_1+A_2$ is invertible. Using the orthogonality of supports condition, this implies that $\supp(A_1) = \supp(A_1+B_1)$ and $\supp(A_2) = \supp(A_2 + B_2)$. 
In addition $(A_1+A_2)^{-1} = A_1^{-1} + A_2^{-1}$. 
Thus we can use~\eqref{eq:geominv} together with the orthogonality of supports condition to compute the mean
\begin{align*}
(A_1+A_2) \sharpf (B_1+B_2) &= (A_1+A_2)^{1/2} f\left((A_1+A_2)^{-1/2} (B_1+B_2) (A_1+A_2)^{-1/2}\right) (A_1+A_2)^{1/2} \\
&= A_1^{1/2} f\left(A_1^{-1/2} B_1 A_1^{-1/2}\right) A_1^{1/2} + A_2^{1/2} f\left(A_2^{-1/2} B_2 A_2^{-1/2}\right) A_2^{1/2} \ .
\end{align*}
Using again~\eqref{eq:geominv} for each one of the two terms, we proved the desired statement when $A_1 + A_2$ is invertible.
 For the general case, we let $P$ be the orthogonal projector onto $\supp(A_1+B_1)$. Then we apply the previous argument to $A_1 + \varepsilon P$ and $A_2 + \varepsilon(I-P)$ for $\eps > 0$, and use the continuity property (iii) to take the limit $\eps \to 0$ and conclude.}

(ix) \change{This follows from the integral representation \eqref{eq:intrepsharpf} and the fact that $\supp((tA):B) = \supp(A) \cap \supp(B)$, which can be easily shown from the variational formulation \eqref{eq:varpar}.}

\sout{After a suitable unitary, we can assume $A$ has the block form \mbox{$A = \begin{sm} A_1 & 0\\ 0 & 0\end{sm}$} where $A_1$ is invertible. In the same basis we write \mbox{$B = \begin{sm} B_{11} & B_{12}\\ B_{12}^* & B_{22} \end{sm}$}. Note that \mbox{$B \leq 2 \begin{sm} B_{11} & 0\\ 0 & B_{22} \end{sm}$} and $A \leq 2A$. Thus by monotonicity we have $A \sharpf B \leq 2 (A_1 \oplus 0) \sharpf (B_{11} \oplus B_{22}) = 2 (A_1 \sharpf B_{11}) \oplus (0 \sharpf B_{22})$. To conclude observe that by our assumptions and using the properties of $\sharpf$, we get $0 \sharpf B_{22} = (\lim_{x \downarrow 0} x f(\frac{1}{x})) I \sharpf B_{22} = 0$ which shows that $A \sharpf B \ll A$. For $B$, we use the same argument inverting the roles of $A$ and $B$, we get $A \sharpf B \leq 2 f(\frac{1}{2}) (A_{11} \oplus A_{22}) \sharpf (B_{1} \oplus 0) = 2 f(\frac{1}{2}) (A_{11} \sharpf B_{1}) \oplus (A_{22} \sharpf 0)$. Then using the fact that $f(0) = 0$, we have $A_{22} \sharpf 0 = 0$ and we get the desired result.}

(x) \change{This follows from (vii) and the continuity of $f$.} \sout{If $A$ is invertible this follows immediately from \eqref{eq:geominv}. If not, we use an argument by blocks like above.}
\end{proof}

\begin{remark}[Lack of continuity of $\sharpf$ in general]
Property (x) is not true if we remove the condition $B_n \ll A$.
Indeed consider a unit vector $v \in \cH$, and let $A = vv^*$ and $B_n = v_n v_n^*$ where $v_n \neq v$ are unit vectors for all $n$ satisfying $v_n \rightarrow v$. Then if $f$ satisfies the conditions in property (ix) in Proposition \ref{prop:p}, we have that $A \sharpf B_n = 0$ for all $n$, and yet $A \sharpf (\lim_n B_n) = vv^*$.
\end{remark}

\change{
We will also need some specific properties that hold for $f(x) = x^{\beta}$ for $\beta \in (0,1)$. For such a function, we write the Kubo-Ando mean $\sharpf$ as $\#_{\beta}$.

\begin{proposition}
For any $\beta \in (0,1)$, we have

(xi) Tensor products: $(A_1 \otimes A_2) \#_{\beta} (B_1 \otimes B_2) = (A_1 \#_{\beta} B_1) \otimes (A_2 \#_{\beta} B_2)$.

(xii) $(aA) \#_{\beta} (bB) = a^{1-\beta} b^{\beta} (A \#_{\beta} B)$ for any $a \geq 0,b\geq 0$.
\end{proposition}
\begin{proof}
(xi) When $A_1 \otimes A_2$ is invertible this follows from the formula \eqref{eq:geominv}. If not, note that for $\eps > 0$ we have
\begin{equation}
\label{eq:additivityeps}
((A_1 + \eps I) \otimes (A_2 + \eps I)) \sharpf (B_1 \otimes B_2) \; = \; ((A_1 + \eps I) \sharpf B_1) \otimes ((A_2 + \eps I) \sharpf B_2).
\end{equation}
When $\eps \downarrow 0$, note that $(A_1 + \eps I) \otimes (A_2 + \eps I) \downarrow A_1 \otimes A_2$. Thus by property (iii) we get the required equality by taking the limit $\eps \downarrow 0$ in \eqref{eq:additivityeps}.

(xii) Using property (ix), if $a = 0$ then both sides of the equality are $0$. Otherwise, if $a > 0$ then using the formula \eqref{eq:geominv}, we have that for any $\eps > 0$, $(a(A+\eps I)) \#_{\beta} (bB) = a^{1-\beta} b^{\beta} ((A+\eps I) \#_{\beta} B)$. Using the continuity property (iii) and taking $\eps \downarrow 0$, we have that $(a(A+\eps I)) \#_{\beta} (bB) \downarrow (aA) \#_{\beta} (bB)$ and $((A+\eps I) \#_{\beta} B) \downarrow (A \#_{\beta} B)$ which establishes the desired statement.
\end{proof}
}

We note that the matrix geometric mean is often defined for positive semidefinite operators as the limit as $\eps \to 0$ of the formula~\eqref{eq:geominv} applied to $A+\eps I$ and $B+\eps I$ and this clearly matches with the general approach of Kubo-Ando. This is the way it is presented in~\cite{bhatia2009positive} and we refer to~\cite{KW20} for a systematic study of the properties of the geometric R\'enyi divergence with this definition.

\section{Properties for positive semidefinite operators}
\label{sec:states}

In this section we state and prove basic properties for the new quantity $\newD_{\alpha}(\rho\|\sigma)$.

\begin{proposition}[First properties]
\label{prop:compact}
The feasible set of \eqref{eq:expr_min_def} is nonempty iff $\rho \ll \sigma$. Furthermore, if $\rho \ll \sigma$ the infimum in \eqref{eq:expr_min_def} is attained at some $A \ll \sigma$, more precisely, at some $A \leq \| \sigma^{-1/2} \rho \sigma^{-1/2} \|_{\infty}^{\alpha-1} \tr(\rho) P$ where $P$ is the projector on $\supp(\sigma)$ and the inverses are generalized inverses.
\end{proposition}
\begin{proof}
If \eqref{eq:expr_min_def} is feasible then $\rho \leq \sigma \#_{1/\alpha} A \ll \sigma$ by Proposition \ref{prop:p}. Conversely, assume $\rho \ll \sigma$ then  $A = \| \sigma^{-1/2} \rho \sigma^{-1/2} \|_{\infty}^{\alpha-1} \rho$ is feasible. In fact, we have
\begin{align*}
\left(\sigma^{-1/2} \rho \sigma^{-1/2} \right)^{\alpha} \leq \left\| \sigma^{-1/2} \rho \sigma^{-1/2} \right\|_{\infty}^{\alpha - 1} \left(\sigma^{-1/2} \rho \sigma^{-1/2} \right) \ .
\end{align*}
Using the operator monotonicity of $t \mapsto t^{1/\alpha}$, we get
\begin{align*}
\rho &\leq \left\| \sigma^{-1/2} \rho \sigma^{-1/2} \right\|_{\infty}^{\frac{\alpha - 1}{\alpha}} \sigma^{1/2} \left(\sigma^{-1/2} \rho \sigma^{-1/2} \right)^{1/\alpha} \sigma^{1/2} \\
&= \sigma \#_{1/\alpha} A \ .
\end{align*}
Note that this shows the program~\eqref{eq:expr_min_def} always has a trivial achievable value of $\| \sigma^{-1/2} \rho \sigma^{-1/2} \|^{\alpha-1}_{\infty} \tr(\rho)$.

Let $P$ be a projector on $\supp(\sigma)$, and let $A$ be a feasible point of \eqref{eq:expr_min_def}. Then $PAP$ is also feasible and satisfies $\tr(PAP) \leq \tr(A)$. Indeed, by the transformer inequality we have
\[
\sigma\#_{1/\alpha} PAP \geq P(\sigma \#_{1/\alpha} A)P \geq P\rho P = \rho \ .
\]
Thus this means we can restrict $A$ to satisfy $A \ll \sigma$. If we in addition assume that $A$ achieves a value for the objective function in~\eqref{eq:expr_min_def} that is at least as good as the trivial value of $\| \sigma^{-1/2} \rho \sigma^{-1/2} \|^{\alpha-1}_{\infty} \tr(\rho)$, we may assume that $\tr(A) \leq \| \sigma^{-1/2} \rho \sigma^{-1/2} \|^{\alpha-1}_{\infty} \tr(\rho)$ which implies that $A \leq \| \sigma^{-1/2} \rho \sigma^{-1/2} \|^{\alpha-1}_{\infty} \tr(\rho) P$. 

\change{We can thus write $\newQ_{\alpha}(\rho \| \sigma) = \inf \left\{ \tr(A) : \sigma \#_{1/\alpha} A \geq \rho \text{ and } A \leq \| \sigma^{-1/2} \rho \sigma^{-1/2} \|^{\alpha-1}_{\infty} \tr(\rho) P\right\}$. The feasible set is closed (by Proposition \ref{prop:p}) and bounded, hence the infimum is attained.}
\sout{Now, to show that the minimum is achieved, consider a sequence $A_n \ll \sigma$ such that $A_n \rightarrow A$. Then by Proposition \ref{prop:p} we have $\rho \leq \sigma \#_{1/\alpha} A_n \rightarrow \sigma \#_{1/\alpha} A$. This means that the limit point $A$ is feasible, and so the minimum is attained.}
\end{proof}

We now show that $\newD_{\alpha}$ satisfies the main properties of a R\'enyi divergence: it satisfies the data-processing inequality and for commuting states, it matches with the classical R\'enyi divergence.
\begin{proposition}
\label{prop:data_processing}
Let $\alpha > 1$. The function $(\rho, \sigma) \mapsto \newQ_{\alpha}(\rho, \sigma)$ is jointly convex. Furthermore  $\newD_{\alpha}$ \sout{satisfies the data-processing inequality} is monotone under trace-preserving positive maps. More precisely, let $\rho$ and $\sigma$ be positive semidefinite operators on the Hilbert space $X$. Then if $\cN$ is a \sout{completely} positive and trace-preserving map from $\Lin(X)$ to $\Lin(Y)$ then
\begin{align*}
\newD_{\alpha}(\cN(\rho) \| \cN(\sigma)) \leq \newD_{\alpha}(\rho \| \sigma) \ .
\end{align*}
In addition, if $\rho$ and $\sigma$ commute, then $\newD_{\alpha}(\rho \| \sigma) = D_{\alpha}(\rho \| \sigma)$.
\end{proposition}
\begin{proof}
Joint convexity follows directly from the joint concavity property of the matrix geometric mean~\eqref{eq:concavitygm}. In fact, for any positive semidefinite operators $\rho_0, \rho_1, \sigma_0, \sigma_1, A_0, A_1$ satisfying $\rho_0 \leq \sigma_0 \#_{1/\alpha} A_0$ and $\rho_1 \leq \sigma_1 \#_{1/\alpha} A_1$, we have for $\lambda \in [0,1]$
\begin{align*}
(1-\lambda) \rho_0 + \lambda \rho_1 
&\leq (1-\lambda)  \sigma_0 \#_{1/\alpha} A_0 + \lambda \sigma_1 \#_{1/\alpha} A_1 \\
&= ((1-\lambda)  \sigma_0) \#_{1/\alpha} ((1-\lambda) A_0) + (\lambda \sigma_1) \#_{1/\alpha} (\lambda A_1) \\
&\leq  \left( (1-\lambda) \sigma_0 + \lambda \sigma_1 \right) \#_{1/\alpha} \left( (1-\lambda) A_0 + \lambda A_1 \right) \ .
\end{align*}
Taking the minimum over $A_0$ and $A_1$, we obtain the desired result.

The data-processing inequality with positive trace-preserving maps follows immediately from the monotonicity property of $\#_{1/\alpha}$ in Proposition \ref{prop:monotonepositive}. In fact, assuming that $\rho \ll \sigma$ (otherwise the statement clearly holds) let $A$ be an optimal point for $\newQ_{\alpha}(\rho\|\sigma)$ so that $\rho \leq \sigma \#_{1/\alpha} A$. Then we have that
\begin{align*}
\cN(\rho) 
&\leq \cN(\sigma \#_{1/\alpha} A) \\
&\leq \cN(\sigma) \#_{1/\alpha} \cN(A) \ ,
\end{align*}
As such $\cN(A)$ is feasible for $\newQ_{\alpha}(\cN(\rho)\|\cN(\sigma))$ with $\tr(\cN(A)) =  \tr(A)$, this proves that $\newQ_{\alpha}(\cN(\rho)\|\cN(\sigma)) \leq \newQ_{\alpha}(\rho\|\sigma)$.
\sout{The data-processing inequality follows from joint-convexity in a generic way (see e.g.,~\cite{Tom15book}) but it is quite straightforward to prove it directly here. In fact, assuming that $\rho \ll \sigma$ (otherwise the statement clearly holds) let $A$ be an optimal point for $\newQ_{\alpha}(\rho\|\sigma)$ so that $\rho \leq \sigma \#_{1/\alpha} A$. Then we have that
where we used the joint-concavity and transformer inequality of the mean $\#_{1/\alpha}$ (see Section~\ref{sec:gm}).  As such $\cN(A)$ is feasible for $\newQ_{\alpha}(\cN(\rho)\|\cN(\sigma))$ with $\tr(\cN(A)) =  \tr(A)$, this proves that $\newQ_{\alpha}(\cN(\rho)\|\cN(\sigma)) \leq \newQ_{\alpha}(\rho\|\sigma)$.}

To analyze the commutative case, consider $\rho$ and $\sigma$ commuting operators. It suffices to assume that $\rho \ll \sigma$ in what follows. To show that $\newD_{\alpha}(\rho \| \sigma) \leq D_{\alpha}(\rho \| \sigma)$ it suffices to take $A = \rho^{\alpha} \sigma^{1-\alpha}$ which commutes with $\rho$ and $\sigma$. Then $\sigma \#_{1/\alpha} A = \sigma^{1-1/\alpha} (\rho^{\alpha} \sigma^{1-\alpha})^{1/\alpha} = \rho$.
To prove $\newD_{\alpha}(\rho \| \sigma) \geq D_{\alpha}(\rho \| \sigma)$ consider a common eigenbasis $\ket{1}, \dots, \ket{d}$ for $\rho$ and $\sigma$ and consider the map $\cM(W) = \sum_{i=1}^d \proj{i} W \proj{i}$. Note that $\cM$ is completely positive and trace-preserving and we have $\cM(\rho) = \rho$ and $\cM(\sigma) = \sigma$. Given an optimal choice of $A$ in the program \eqref{eq:expr_min_def} for $\rho$ and $\sigma$, we can write as before
\begin{align*}
\rho = \cM(\rho) 
&\leq \cM(\sigma \#_{1/\alpha} A) \\
&\leq \cM(\sigma) \#_{1/\alpha} \cM(A) \\
&= \sigma \#_{1/\alpha} \cM(A) \ .
\end{align*}
Noting that $\tr(\cM(A)) = \tr(A)$, we have constructed another optimal solution where the matrix $A$ commutes with $\rho$ and $\sigma$. In this case $\sigma \#_{1/\alpha} A = \sigma^{1-1/\alpha}  A^{1/\alpha}$. Thus, the condition $\rho \leq \sigma \#_{1/\alpha} A$ translates to $\sigma^{1/\alpha-1} \rho \leq A^{1/\alpha}$. As all the matrices are diagonal in the same basis, we can take both sides of this inequality to the power $\alpha$ and get $\rho^{\alpha} \sigma^{1-\alpha} \leq A$.
Taking the trace, we get that $D_{\alpha}(\rho \| \sigma) \leq \newD_{\alpha}( \rho \| \sigma)$.
\end{proof}

\begin{remark}[Non-monotonicity in $\alpha$]
We would like to emphasize that $\newD_{\alpha}$ is \emph{not} monotone in $\alpha$. This is illustrated in Figure~\ref{fig:exampleDsharp} where we take $\rho = \proj{\phi_\eps}$ with $\ket{\phi_\eps} = \sqrt{\eps} \ket{00} + \sqrt{1-\eps} \ket{11}$ and $\sigma = \id \otimes (\eps \proj{0} + (1-\eps) \proj{1})$.
\end{remark}
\begin{figure}[ht]
  \centering
%
%
\definecolor{mycolor1}{rgb}{0.00000,0.44700,0.74100}%
\definecolor{mycolor2}{rgb}{0.85000,0.32500,0.09800}%
\definecolor{mycolor3}{rgb}{0.92900,0.69400,0.12500}%
\definecolor{mycolor4}{rgb}{0.49400,0.18400,0.55600}%
\pgfplotstableread{
   eps  Dhat  Dtilde  Dsharp
   0.000010000000000   1.000000000000000   0.001979612414616   0.064948649461560
   0.000012890103671   1.000000000000000   0.002341569450270   0.071441231473763
   0.000016615477265   1.000000000000000   0.002769366330080   0.078550560264805
   0.000021417522448   1.000000000000000   0.003274877638815   0.086329262360317
   0.000027607408473   1.000000000000000   0.003872088981550   0.094833289874017
   0.000035586235731   1.000000000000000   0.004577462903480   0.104122012789102
   0.000045871026782   1.000000000000000   0.005410365310150   0.114258043980299
   0.000059128229072   1.000000000000000   0.006393561448295   0.125307067064392
   0.000076216900261   1.000000000000000   0.007553791545419   0.137337574235710
   0.000098244374584   1.000000000000000   0.008922437244655   0.150420679395882
   0.000126638017347   1.000000000000000   0.010536290953114   0.164629215112980
   0.000163237717228   1.000000000000000   0.012438441062730   0.180037612870245
   0.000210415109807   1.000000000000000   0.014679286581729   0.196720624869649
   0.000271227257933   1.000000000000000   0.017317694862295   0.214752379777029
   0.000349614747313   1.000000000000000   0.020422315591094   0.234205092144158
   0.000450657033775   1.000000000000000   0.024073062704448   0.255147361900893
   0.000580901588538   1.000000000000000   0.028362772969355   0.277642329707312
   0.000748788169884   1.000000000000000   0.033399045064540   0.301745317795258
   0.000965195713735   1.000000000000000   0.039306255352663   0.327500956788360
   0.001244147281276   1.000000000000000   0.046227735182086   0.354940330922695
   0.001603718743751   1.000000000000000   0.054328078243017   0.384077276336132
   0.002067210086592   1.000000000000000   0.063795523638343   0.414904146193352
   0.002664655232571   1.000000000000000   0.074844328912364   0.447387454736323
   0.003434768219505   1.000000000000000   0.087717004815466   0.481462598333634
   0.004427451843494   1.000000000000000   0.102686226974711   0.517028716785910
   0.005707031326057   1.000000000000000   0.120056165059207   0.553942609752306
   0.007356422544596   1.000000000000000   0.140162872692084   0.592012684517047
   0.009482504924681   1.000000000000000   0.163373255166526   0.630992710529874
   0.012223047153898   1.000000000000000   0.190081968858065   0.670575651286053
   0.015755634498807   1.000000000000000   0.220705394625557   0.710387905884904
   0.020309176209047   1.000000000000000   0.255671550019382   0.749984393050027
   0.026178738680524   1.000000000000000   0.295404433208446   0.788845191598577
   0.033744665556542   1.000000000000000   0.340300775449993   0.826374852556804
   0.043497223736375   1.000000000000000   0.390696427902371   0.861906211043341
   0.056068372335831   1.000000000000000   0.446818451282356   0.894711751860448
   0.072272713206762   1.000000000000000   0.508717080199561   0.924027744461401
   0.093160276581255   1.000000000000000   0.576168429788131   0.949099672925276
   0.120084562314231   1.000000000000000   0.648532706613235   0.969262967224735
   0.154790245750538   1.000000000000000   0.724540716718778   0.984076983077784
   0.199526231496888   1.000000000000000   0.801956430936513   0.993528519561780
}\datatable

\newcommand{\drawDsharpPlot}[1]{%
\begin{tikzpicture}
\begin{axis}[%
width=2in,
height=1.25in,
at={(1.011in,0.642in)},
scale only axis,
xmin=0,
xmax=0.2,
ymin=0,
ymax=1.1,
ylabel style={rotate=-90},
xlabel={$\eps$},
ylabel={},
xtick={0,0.05,0.1,0.15,0.2},
xticklabel style={/pgf/number format/fixed},
title={$\alpha=3/2$},
axis background/.style={fill=white},
legend style={legend cell align=left, align=left, draw=white!15!black}
]
\addplot [dotted,  color=black,  line width=1.0pt] table[x=eps, y=Dhat]{\datatable} node[below,pos=0.1] {$\widehat{D}$};
\addplot [smooth,color=mycolor1, dashed, line width=1.0pt] table[x=eps, y=Dtilde]{\datatable} node[above,pos=0.8] {$\widetilde{D}$};
\ifnum#1=1
\addplot [smooth,color=mycolor2, line width=2.0pt] table[x=eps, y=Dsharp]{\datatable} node[below,pos=0.8] {$\newD$};
\fi
\end{axis}
\end{tikzpicture}%
}
%
%
\definecolor{mycolor1}{rgb}{0.00000,0.44700,0.74100}%
\definecolor{mycolor2}{rgb}{0.85000,0.32500,0.09800}%
\definecolor{mycolor3}{rgb}{0.92900,0.69400,0.12500}%
\definecolor{mycolor4}{rgb}{0.49400,0.18400,0.55600}%
\begin{tikzpicture}

\begin{axis}[%
width=2in,
height=1.25in,
at={(1.011in,0.642in)},
scale only axis,
xmin=1,
xmax=4,
ymin=0,
ymax=1.1,
ylabel style={rotate=-90},
xlabel={$\alpha$},
ylabel={$D$},
xtick={1,1.5,2,2.5,3,3.5,4},
title={$\varepsilon=10^{-3}$},
axis background/.style={fill=white},
legend style={legend cell align=left, align=left, draw=white!15!black}
]

\addplot [color=black, line width=1.0pt,dotted,forget plot]
  table[row sep=crcr]{%
4	1\\
1.03225806451613	1\\
} node[below,pos=0.5] {$\widehat{D}$};


%



%

\addplot [smooth,color=mycolor2, line width=2.0pt,forget plot]
  table[row sep=crcr]{%
4	0.523876798561143\\
2	0.362156516791363\\
1.5	0.331227634402546\\
1.33333333333333	0.334933884719728\\
1.14285714285714	0.383511520071417\\
1.06666666666667	0.465506430370875\\
1.03225806451613	0.569805156248011\\
} node[above,pos=0.5] {$\newD$};

\addplot [smooth,color=mycolor1, dashed, line width=1.0pt,forget plot]
  table[row sep=crcr]{%
4	0.314429915374535\\
2	0.088431901576697\\
1.5	0.0402076371440934\\
1.33333333333333	0.0280545512793117\\
1.14285714285714	0.017256948247557\\
1.06666666666667	0.0139010211354024\\
1.03225806451613	0.0125634771586205\\
} node[above,pos=0.5] {$\widetilde{D}$};



%

\end{axis}
\end{tikzpicture}%
  \qquad
  \drawDsharpPlot{1}
  \caption{Left: The solid line corresponds to $\newD_{\alpha}(\rho_{XY}\| \id_X \otimes \rho_Y)$ as a function of $\alpha$ where $\rho_{XY} = \proj{\phi_{\eps}}$ with $\ket{\phi_\eps} = \sqrt{\eps} \ket{00} + \sqrt{1-\eps} \ket{11}$.
  The dashed lines correspond to $\widetilde{D}_{\alpha}(\rho_{XY}\| \id_X \otimes \rho_Y)$. Note that in this case $\widehat{D}_{\alpha}(\rho_{XY}\| \id_X \otimes \rho_Y) = D_{\max}(\rho_{XY}\| \id_X \otimes \rho_Y) = 1$ for all $\alpha \in (1,\infty)$ and all $\eps \in (0,1)$ (this corresponds to the black dotted line at the top).
  \change{Right: The three divergences (for $\alpha = 3/2$) as a function of $\eps \in (0,0.2)$. Note that $\newD$ and $\widetilde{D}$ converge to zero as $\eps \to 0$, whereas this is not the case for $\widehat{D}$.}  The computations were performed using the package~\cite{cvxquad}.
  }
  \label{fig:exampleDsharp}
\end{figure}

Now we turn to properties of the divergence for tensor products and the relation to other quantum R\'enyi divergences.
\begin{proposition}
\label{prop:relation_other_div}
Let $\alpha \in (1,\infty)$. The quantity $\newD_{\alpha}$ is subadditive under tensor products, i.e., for $\rho_1, \sigma_1 \in \Pos(\cH_1)$ and $\rho_2, \sigma_2 \in \Pos(\cH_2)$, we have 
\begin{align}
\label{eq:subadditive_states}
\newD_{\alpha}(\rho_1 \otimes \rho_2 \| \sigma_1 \otimes \sigma_2) \leq \newD_{\alpha}(\rho_1 \| \sigma_1) + \newD_{\alpha}( \rho_2 \| \sigma_2) \ .
\end{align}
In addition, for $\rho, \sigma \in \Pos(\cH)$ it can be related to the measured R\'enyi divergence as follows:
\begin{align}
\label{eq:relation_dmeas}
D^{\mathbb{M}}_{\alpha}(\rho \| \sigma) \leq \newD_{\alpha}(\rho \| \sigma) \leq D^{\mathbb{M}}_{\alpha}(\rho \| \sigma) + \frac{\alpha}{\alpha-1} \log |\mathrm{spec}(\sigma)| \ .
\end{align}
As a consequence,
\begin{align}
\label{eq:reg_im_sand}
\lim_{n \to \infty} \frac{1}{n}  D^{\mathbb{M}}_{\alpha}(\rho^{\otimes n} \| \sigma^{\otimes n}) = \widetilde{D}_{\alpha}(\rho \| \sigma) =  \lim_{n \to \infty} \frac{1}{n} \newD_{\alpha}(\rho^{\otimes n} \| \sigma^{\otimes n})
\end{align}
and 
\begin{align}
D^{\mathbb{M}}_{\alpha}(\rho \| \sigma) \leq \widetilde{D}_{\alpha}(\rho \| \sigma) \leq \newD_{\alpha}(\rho \| \sigma) \leq \widehat{D}_{\alpha}(\rho \| \sigma) \ .
\end{align}
\change{Furthermore $\newD_{\alpha}(\rho \| \sigma) \rightarrow D_{\max}(\rho \| \sigma)$ when $\alpha\rightarrow \infty$.}
\end{proposition}
\begin{proof}
In order to show subadditivity, we may assume $\rho_1 \ll \sigma_1$ and $\rho_2 \ll \sigma_2$ (the statement clearly holds otherwise). For $b \in \{1,2\}$, let $A_{b}$ be a feasible solution for the program \eqref{eq:expr_min_def} for $\newD_{\alpha}(\rho_b \| \sigma_b)$. Then define $A_{12} = A_{1} \otimes A_{2}$. Then, using the tensor product property in Proposition~\ref{prop:p}, we get
\begin{align*}
 (\sigma_1 \otimes \sigma_2) \#_{1/\alpha} A_{12}  &= (\sigma_1 \#_{1/\alpha} A_1) \otimes  (\sigma_{2} \#_{1/\alpha} A_2) \\
&\geq \rho_1 \otimes \rho_2 \ .
\end{align*}
In addition, $\tr(A_{12}) = \tr(A_{1}) \tr(A_{2})$ and we thus get $\newQ_{\alpha}(\rho_1 \otimes \rho_2 \| \sigma_1 \otimes \sigma_2) \leq \newQ_{\alpha}(\rho_1 \| \sigma_1) \newQ_{\alpha}(\rho_2 \| \sigma_2)$ which proves subadditivity.

Next, to relate the different divergences, note that all of these divergences are finite if and only if $\rho \ll \sigma$. So we focus on the case $\rho \ll \sigma$. To show~\eqref{eq:relation_dmeas}, it follows from the data-processing inequality that for any choice of orthonormal basis $\{\ket{v_x}\}_{x}$, the measurement map $\cM(W) = \sum_{x} \proj{v_x} W \proj{v_x}$ satisfies $\newD_{\alpha}(\cM(\rho) \| \cM(\sigma)) \leq \newD_{\alpha}(\rho \| \sigma)$. But as $\cM(\rho)$ and $\cM(\sigma)$ commute, $\newD_{\alpha}(\cM(\rho) \| \cM(\sigma)) = D_{\alpha}(\cM(\rho) \| \cM(\sigma))$. Taking the supremum over measurements $\cM$, we have $D^{\mathbb{M}}_{\alpha}(\rho \| \sigma) \leq \newD_{\alpha}(\rho \| \sigma)$.

In the other direction, consider the pinching map $\cP_{\sigma}$ defined by $\cP_{\sigma}(W) = \sum_{\lambda \in \mathrm{spec}(\sigma)} \Pi_{\lambda} W \Pi_{\lambda}$, where $\mathrm{spec}(\sigma)$ denotes the set of eigenvalues of $\sigma$ and $\Pi_{\lambda}$ is the projector onto the eigenspace of $\lambda$. Consider an optimal solution $A$ of~\eqref{eq:expr_min_def} for the states $\cP_{\sigma}(\rho)$ and $\sigma$. Then we have $\cP_{\sigma}(\rho) \leq \sigma  \#_{1/\alpha} A$. Using the pinching inequality $\rho \leq |\mathrm{spec}(\sigma)| \cP_{\sigma}(\rho)$ \change{(see e.g., \cite{hayashi2002optimal} or \cite[Chapter 2]{Tom15book})}, we obtain
\begin{align*}
\rho \leq |\mathrm{spec}(\sigma)| \cP_{\sigma}(\rho) &\leq |\mathrm{spec}(\sigma)| \cdot \sigma \#_{1/\alpha} A \\
&= \sigma \#_{1/\alpha} (|\mathrm{spec}(\sigma)|^{\alpha} A) \ .
\end{align*}
As such $|\mathrm{spec}(\sigma)|^{\alpha} A$ is feasible for the optimization program~\eqref{eq:expr_min_def} for $\rho$ and $\sigma$ and thus $\newQ_{\alpha}(\rho \| \sigma) \leq |\mathrm{spec}(\sigma)|^{\alpha} \newQ_{\alpha}(\cP_{\sigma}(\rho) \| \sigma)$. But note that $\cP_{\sigma}(\sigma) = \sigma$ and it commutes with $\cP_{\sigma}(\rho)$, so using Proposition~\ref{prop:data_processing}, we have $\newD_{\alpha}(\cP_{\sigma}(\rho) \| \sigma) = D_{\alpha}(\cP_{\sigma}(\rho) \| \cP_{\sigma}(\sigma)) \leq D^{\mathbb{M}}_{\alpha}(\rho \| \sigma)$. Putting everything together, we get
\begin{align*}
\newD_{\alpha}(\rho \| \sigma) \leq D^{\mathbb{M}}_{\alpha}(\rho \| \sigma) + \frac{\alpha}{\alpha-1} \log |\mathrm{spec}(\sigma)| \ .
\end{align*}
 
As a result, we have that for any integer $n \geq 1$,
\begin{align*}
\frac{1}{n} D^{\mathbb{M}}_{\alpha}(\rho^{\otimes n} \| \sigma^{\otimes n}) \leq \frac{1}{n} \newD_{\alpha}(\rho^{\otimes n} \| \sigma^{\otimes n}) \leq \frac{1}{n} D^{\mathbb{M}}_{\alpha}(\rho^{\otimes n} \| \sigma^{\otimes n}) +  
\frac{1}{n} \frac{\alpha}{\alpha-1} \log |\mathrm{spec}(\sigma^{\otimes n})| .
\end{align*}
But it is well-known that $|\mathrm{spec}(\sigma^{\otimes n})|  \leq (n+1)^{\dim \cH-1}$. This shows that the regularization of both $\newD_{\alpha}$ and $D^{\mathbb{M}}_{\alpha}$ give the same value\sout{, the former converging from above and the latter from below}. But the regularization of the measured R\'enyi divergence is known to be equal to $\widetilde{D}_{\alpha}$~\change{\cite{mosonyi2015quantum} (see also \cite[Theorem 4.1]{Tom15book})}. \change{The inequality $\widetilde{D}_{\alpha}(\rho \| \sigma) \leq \newD_{\alpha}(\rho \| \sigma)$ then follows directly because, by subadditivity of $\newD_{\alpha}$ we have
\[
\widetilde{D}_{\alpha}(\rho \| \sigma) = \lim_{n\rightarrow\infty} \frac{1}{n} \newD_{\alpha}(\rho^{\otimes n} \| \sigma^{\otimes n}) = \inf_{n \in \NN} \frac{1}{n} \newD_{\alpha}(\rho^{\otimes n} \| \sigma^{\otimes n}) \leq \newD_{\alpha}(\rho \| \sigma).
\]
}

We now show that $\newD_{\alpha}( \rho \| \sigma) \leq \widehat{D}_{\alpha}( \rho \| \sigma)$. 
Recall that $\widehat{D}_{\alpha}(\rho \| \sigma) = \frac{1}{\alpha - 1} \log \tr \left( \sigma \#_{\alpha} \rho \right)$, where, for $\rho \ll \sigma$, $\sigma \#_{\alpha} \rho$ is defined by
\[
\sigma \#_{\alpha} \rho = \sigma^{1/2} \left( \sigma^{-1/2} \rho \sigma^{-1/2}\right)^{\alpha} \sigma^{1/2} \ ,
\]
where generalized inverses are used if $\sigma$ is not invertible.
If we choose $A = \sigma \#_{\alpha} \rho$, it is immediate to verify that
\begin{align*}
\sigma \#_{1/\alpha} A = \sigma \#_{1/\alpha} ( \sigma \#_{\alpha} \rho) = \rho
\end{align*}
which means that $A$ is feasible for \eqref{eq:expr_min_def}, and so $\newD_{\alpha}(\rho \| \sigma) \leq \widehat{D}_{\alpha}(\rho \| \sigma)$. \change{We note that for $\alpha \in (1,2]$ this inequality actually follows immediately from the fact that $\widehat{D}_{\alpha}$ is the maximal R\'enyi divergence satisfying the data-processing inequality~\cite{Mat13}. As illustrated in Figure~\ref{fig:exampleDsharp}, there can be a large gap between $\newD_{\alpha}$ and $\widehat{D}_{\alpha}$.}

\change{
Finally we prove that $\newD_{\alpha}( \rho \| \sigma) \rightarrow D_{\max}(\rho \| \sigma)$ when $\alpha \rightarrow \infty$. We saw in Proposition \ref{prop:compact} that $\newQ_{\alpha}(\rho \| \sigma) \leq \left\|\sigma^{-1/2} \rho \sigma^{-1/2}\right\|_{\infty}^{\alpha-1} \tr(\rho) = 2^{D_{\max}(\rho \| \sigma)(\alpha-1)} \tr(\rho)$, which implies that $\newD_{\alpha}(\rho \| \sigma) \leq D_{\max}(\rho \| \sigma) + \frac{1}{\alpha - 1} \log \tr(\rho)$. On the other hand it is known (see e.g., \cite[Proposition 4]{MDSFT13}) that $\widetilde{D}_{\alpha}(\rho \| \sigma) \rightarrow D_{\max}(\rho \| \sigma)$ when $\alpha\rightarrow \infty$ thus we get the desired result.
}

%
%
%
\end{proof}


\change{
We conclude this section by establishing additional useful properties of $\newQ_{\alpha}$.
\begin{proposition}[Additional properties for $\newQ_{\alpha}$]
\label{prop:add_props}
For any $\alpha > 1$, the following hold.
(a) Isometric invariance: if $V$ is an isometry, i.e., $V^* V = I$ then $\newQ_{\alpha}(V\rho V^* \| V \sigma V^*) = \newQ_{\alpha}(\rho\| \sigma)$. (b) Positive homogeneity: $\newQ_{\alpha}(\lambda \rho \| \lambda \sigma) = \lambda \newQ_{\alpha}(\rho \| \sigma)$ for $\lambda \geq 0$. (c) Block additivity: $\newQ_{\alpha}(\rho_1 + \rho_2 \| \sigma_1 + \sigma_2) = \newQ_{\alpha}(\rho_1 \| \sigma_1) + \newQ_{\alpha}(\rho_2 \| \sigma_2)$ provided $\supp(\rho_1+\sigma_1) \perp \supp(\rho_2+\sigma_2)$.
\end{proposition}
\begin{proof}
(a) The inequality $\newQ_{\alpha}(V\rho V^* \| V\sigma V^*) \leq \newQ_{\alpha}(\rho \| \sigma)$ follows from the data-processing inequality. For the other inequality if $A$ is an optimal point for $\newQ_{\alpha}(V \rho V^* \| V \sigma V^*)$, using the transformer inequality, $V^* A V$ is feasible point for the program defining $\newQ_{\alpha}(\rho \| \sigma )$. But as $VV^* \leq \id$, we have $\tr(V^* A V) = \tr(A V V^*) \leq \tr(A)$, which proves the desired result. \\
(b) Immediate to verify from the definition of $\newQ_{\alpha}(\rho \| \sigma)$.\\
(c) We use the fact that $(A_1 + A_2)\#_{1/\alpha} (B_1 + B_2) = (A_1 \#_{1/\alpha} B_1) + (A_2 \#_{1/\alpha} B_2)$ established in Section \ref{sec:gm}. Let $P_1$, $P_2$ be projectors on $\supp(\rho_1 + \sigma_1)$ and $\supp(\rho_2 + \sigma_2)$ respectively. For the inequality $\newQ_{\alpha}(\rho_1 + \rho_2 \| \sigma_1 + \sigma_2) \leq \newQ_{\alpha}(\rho_1 \| \sigma_1) + \newQ_{\alpha}(\rho_2 \| \sigma_2)$, consider $A_1$ and $A_2$ two optimal points for $\newQ_{\alpha}(\rho_1 \| \sigma_1)$ and $\newQ_{\alpha}(\rho_2 \| \sigma_2)$, respectively. Using the transformer inequality and the orthogonality condition, note that $P_1 A_1 P_1$ and $P_2 A_2 P_2$ are also optimal points for $\newQ_{\alpha}(\rho_1 \| \sigma_1)$ and $\newQ_{\alpha}(\rho_2 \| \sigma_2)$. Using the property of the geometric mean for orthogonal sums of operators, we get that $P_1 A_1 P_1 + P_2 A_2 P_2$ is feasible for the program defining $\newQ_{\alpha}(\rho_1 + \rho_2 \| \sigma_1 + \sigma_2)$. \\
For the reverse inequality. Let $A$ be an optimal solution in the definition of $\newQ_{\alpha}(\rho_1 + \rho_2 \| \sigma_1 + \sigma_2)$ so that $(\sigma_1 + \sigma_2) \#_{1/\alpha} A \geq \rho_1 + \rho_2$. 
Let $\pi(A) = P_1 A P_1 + P_2 A P_2$ be the projection map on ``block-diagonal'' operators. Then, letting $\sigma = \sigma_1 + \sigma_2$ and $\rho = \rho_1 + \rho_2$ we have
\[
\sigma \#_{1/\alpha} \pi(A) = \pi(\sigma) \#_{1/\alpha} \pi(A) \geq \pi(\sigma \#_{1/\alpha} A) \geq \pi(\rho) = \rho.
\]
Since $\tr \pi(A) = \tr A$, it follows that $\pi(A)$ is also an optimal solution in the definition of $\newQ_{\alpha}(\rho_1 + \rho_2\| \sigma_1 + \sigma_2)$. Call $A_1 = P_1 A P_1$ and $A_2 = P_2 A P_2$. By the previous equation and the block additivity of the geometric mean, we see that $\sigma_1 \#_{1/\alpha} A_1 \geq \rho_1$ and $\sigma_2 \#_{1/\alpha_2} A_2 \geq \rho_2$. This implies that $\newQ_{\alpha}(\rho_1 \| \sigma_1) \leq \tr A_1$ and $\newQ_{\alpha}(\rho_2 \| \sigma_2) \leq \tr A_2$, which in turn implies $\newQ_{\alpha}(\rho_1 \| \sigma_1) + \newQ_{\alpha}(\rho_2 \| \sigma_2) \leq \tr \pi(A) = \newQ_{\alpha}(\rho \| \sigma)$ as desired.
\end{proof}
An immediate consequence of the proposition above is that if $\ket{\psi}$ is a unit normed vector in $\cH$ then $\newQ_{\alpha}(\rho \otimes \proj{\psi} \| \sigma \otimes \proj{\psi}) = \newQ_{\alpha}(\rho \| \sigma)$, because the map $V\ket{\phi} = \ket{\phi} \otimes \ket{\psi}$ is an isometry.
In turn this implies the following useful property for classical-quantum states.

\begin{proposition}
Let $\Sigma$ be a finite set, $p(x) \geq 0$, $\rho(x)$ and $\sigma(x) \in \Pos(\cH)$ for all $x \in \Sigma$. Then, $\newQ_{\alpha}$ has the direct-sum property for classical-quantum states:
\begin{align*}
\newQ_{\alpha}\left(\sum_{x \in \Sigma} p(x) \proj{x} \otimes \rho(x) \Big\| \sum_{x \in \Sigma} p(x) \proj{x} \otimes \sigma(x)\right) &= \sum_{x \in \Sigma} p(x)  \newQ_{\alpha}(\rho(x) \| \sigma(x)) \ .
\end{align*}
\end{proposition}
}

\sout{\emph{Proof.} \quad To show $\leq$, we consider feasible points $A(x)$ for the optimization programs~\eqref{eq:expr_min_def} for $\rho(x)$ and $\sigma(x)$ and construct $A = \sum_{x} p(x) \proj{x} \otimes A(x)$. By the direct sum property in Proposition~\ref{prop:p}, $A$ is feasible for the optimization program~\eqref{eq:expr_min_def} for $\rho := \sum_{x} p(x) \proj{x} \otimes \rho(x)$ and $\sigma := \sum_{x} p(x) \proj{x} \otimes \sigma(x)$. In addition, $\tr(A) = \sum_{x} p(x) \tr(A(x))$ which proves the desired inequality.
For the other direction $\geq$, consider a feasible point $A$  for the optimization programs~\eqref{eq:expr_min_def} for $\rho$ and $\sigma$. By definition, we have $\rho \leq \sigma \#_{1/\alpha} A$. We apply the completely positive map $\cM(W) = \sum_{x \in \Sigma} (\proj{x} \otimes I) W (\proj{x} \otimes I)$ to this inequality, followed by the joint-concavity and the transformer inequality (see Section~\ref{sec:gm}), we get $\rho \leq \sigma  \#_{1/\alpha} \cM(A)$. Now $\cM(A)$ has a block diagonal form and we let $p(x) A(x)$ be the block labelled $x$. By the direct sum property in Proposition~\ref{prop:p}, we have $p(x) \rho(x) \leq (p(x) \sigma(x))  \#_{1/\alpha} (p(x) A(x))$ and thus, $A(x)$ is feasible for the optimization program of $\newD_{\alpha}(\rho(x) \| \sigma(x))$ and $\tr(A) = \sum_{x} p(x) \tr(A(x))$ which leads to the desired inequality. \qed}

\section{Properties for positive maps}
\label{sec:channels}

The notion of divergence between states can naturally be extended to a divergence between channels by maximizing over the input states. Here we consider the stabilized version where a reference system that is unaffected by the channels is allowed.
Let $X,X', Y$ be Hilbert spaces with $\dim X = \dim X'$ and $\cN, \cM$ be completely positive maps from $\Lin(X')$ to $\Lin(Y)$. It will be convenient to write the definition of the channel divergence in terms of the Choi states of the channels. For this, we define $\Phi_{XX'}$ as an unnormalized maximally entangled state of the form $\Phi_{XX'} = \sum_{x,x'} \ketbra{x}{x'}_{X} \otimes \ketbra{x}{x'}_{X'}$, where $\{\ket{x}\}_{x}$ labels a fixed basis of $X$ and $X'$. Then we let
 $J^{\cN}_{XY} = (\cI_{X} \otimes \cN_{X' \to Y})(\Phi_{XX'})$ and $J^{\cM}_{XY} = (\cI_{X} \otimes \cM_{X' \to Y})(\Phi_{XX'})$ be the Choi matrices of these channels. Here $\cI_{X}$ denotes the identity map on $\Lin(X)$. Observe that for any density operator $\omega \in \D(X)$, $\omega_{X}^{\frac{1}{2}} J^{\cN}_{XY} \omega_{X}^{\frac{1}{2}} = (\cI_{X} \otimes \cN_{X' \to Y})(\Omega_{XX'})$ where $\Omega_{XX'}$ is the pure state $\omega_{X}^{\frac{1}{2}} \Phi_{XX'} \omega_{X}^{\frac{1}{2}}$. For any divergence $\mathbf{D}$, the corresponding channel divergence is defined as:
\begin{align}
\label{eq:def_channel_div}
\mathbf{D}(\cN \| \cM) := \sup_{\omega_{X} \in \D(X)} \mathbf{D}(\omega_{X}^{\frac{1}{2}} J_{XY}^{\cN} \omega_{X}^{\frac{1}{2}} \| \omega_{X}^{\frac{1}{2}} J_{XY}^{\cM} \omega_{X}^{\frac{1}{2}}) \ .
\end{align}
We refer to~\cite{LKDW18} for a more detailed discussion of this definition. For $\mathbf{D} = \newD_{\alpha}$, our first result is an expression for the channel divergence in terms of a convex optimization program.
\begin{theorem}
\label{thm:channel_div}
For any $\alpha \in (1,\infty)$ and completely positive maps $\cN$ and $\cM$, we can write
\begin{align}
\label{eq:expr_channel_D}
\newD_{\alpha}(\cN \| \cM) &= \frac{1}{\alpha-1} \log \newQ_{\alpha}(\cN \| \cM) \\
\label{eq:expr_channel_Q}
\newQ_{\alpha}(\cN \| \cM) 
&= \inf_{A_{XY} \geq 0} \| \tr_{Y}(A_{XY}) \|_{\infty}  \quad \textup{ s.t. } \quad  J^{\cN}_{XY} \leq J^{\cM}_{XY} \#_{1/\alpha} A_{XY} \ ,
\end{align}
where $\| . \|_{\infty}$ denotes the operator norm.
\end{theorem}
\begin{remark}
Note that the constraint in~\eqref{eq:expr_channel_Q} is jointly convex in $J^{\cN}_{XY}$ and $J^{\cM}_{XY}$.
\end{remark}
\begin{proof}
First, if $J_{XY}^{\cN} \ll J_{XY}^{\cM}$ is not satisfied, then both quantities are $\infty$: we can take $\omega_{X} = \frac{\id}{\dim X}$ in \eqref{eq:def_channel_div} and have $\newD_{\alpha}({\cN} \| {\cM}) = \infty$, and on the other hand the optimization program in \eqref{eq:expr_channel_Q} is infeasible. So we may assume that $J_{XY}^{\cN} \ll J_{XY}^{\cM}$ in what follows.
We have to show that
\begin{equation}
\label{eq:goalchannel}
\sup_{\omega_{X} \in \D(X)} \newQ_{\alpha}(\omega_{X}^{\frac{1}{2}} J_{XY}^{\cN} \omega_{X}^{\frac{1}{2}} \| \omega_{X}^{\frac{1}{2}} J_{XY}^{\cM} \omega_{X}^{\frac{1}{2}})
\;
=
\;
\min_{A_{XY} \geq 0} \| \tr_{Y}(A_{XY}) \|_{\infty} \;\; \text{ s.t. } \;\; J^{\cN}_{XY} \leq J^{\cM}_{XY} \#_{1/\alpha} A_{XY}.
\end{equation}

We will start by showing that, if we strengthen the condition $\omega_{X} \in \D(X)$ with $\omega_{X} \in \D(X), \omega_{X} > 0$ in the left-hand side, equality holds. We will later show that the condition $\omega_{X} > 0$ leads to the same quantity. With this additional condition, the left-hand side of \eqref{eq:goalchannel} is
\begin{align}
\label{eq:proof_channel_intermediate}
\sup_{\substack{\omega_{X} \in \D(X) \\ \omega_{X} > 0}} & \min_{B_{XY} \geq 0} \tr(B_{XY})  \\
& \quad \omega_{X}^{\frac{1}{2}} J^{\cN}_{XY} \omega_{X}^{\frac{1}{2}} \leq (\omega_{X}^{\frac{1}{2}} J^{\cM}_{XY} \omega_{X}^{\frac{1}{2}}) \#_{1/\alpha} B_{XY}  \ . \label{eq:proof_channel_constraint_1}
\end{align}
Now using the fact that $\omega_{X}$ is invertible together with the transformer inequality, we get
\begin{align*}
(\omega_{X}^{\frac{1}{2}} J^{\cM}_{XY} \omega_{X}^{\frac{1}{2}}) \#_{1/\alpha} B_{XY}
&= \omega_{X}^{\frac{1}{2}} \left(  J^{\cM}_{XY}  \#_{1/\alpha} (\omega_{X}^{-\frac{1}{2}} B_{XY} \omega_{X}^{-\frac{1}{2}} ) \right) \omega_{X}^{\frac{1}{2}}
\end{align*}
Thus, the constraint in~\eqref{eq:proof_channel_constraint_1} is equivalent to 
\begin{align*}
 J^{\cN}_{XY} \leq J^{\cM}_{XY} \#_{1/\alpha} (\omega_{X}^{-\frac{1}{2}} B_{XY} \omega_{X}^{-\frac{1}{2}} )  \ .
 \end{align*}
 Thus, by performing a change of variable $A_{XY} = \omega_{X}^{-\frac{1}{2}} B_{XY} \omega_{X}^{-\frac{1}{2}}$, the program in~\eqref{eq:proof_channel_intermediate} becomes
 \begin{align*}
\sup_{\substack{\omega_{X} \in \D(X) \\ \omega_{X} > 0}} & \min_{A_{XY} \geq 0} \tr \left( \omega_{X}^{\frac{1}{2}} A_{XY}  \omega_{X}^{\frac{1}{2}} \right)  \\ 
& \quad  J^{\cN}_{XY}  \leq J^{\cM}_{XY} \#_{1/\alpha} A_{XY}  \ .
\end{align*}
Now using Sion's minmax theorem, observe that we can exchange the minimization and the maximization. In fact, the objective function is linear in both $\omega_{X}$ and in $A_{XY}$, and the set of invertible density operators is convex. In addition, as we assumed that $J^{\cN}_{XY}  \ll J^{\cM}_{XY}$, using Proposition~\ref{prop:compact}, we may restrict the set of $A_{XY}$ we optimize over to be convex and compact. To conclude, it suffices to observe that
$\sup_{\substack{\omega_{X} \in \D(X) \\ \omega_{X} > 0}} \tr(\omega_{X} A_{XY}) = \| \tr_{Y} A_{XY} \|_{\infty}$.


Since replacing the condition $\omega_X \geq 0$ by $\omega_X > 0$ can only decrease the LHS of \eqref{eq:goalchannel}, we have shown the direction $\geq$ of \eqref{eq:goalchannel}. It thus remains to show the direction $\leq$. Take an optimal feasible solution $A_{XY}$ of~\eqref{eq:expr_channel_Q} and let us write $\lambda$ for its value. Now consider an $\omega_{X} \in \D(X)$ and define $A_{XY}^{\omega} = \omega_{X}^{\frac{1}{2}} A_{XY} \omega_{X}^{\frac{1}{2}}$. By construction $\tr(A^{\omega}_{XY}) \leq \lambda$. In addition, we have
\begin{align*}
\omega_{X}^{\frac{1}{2}} J^{\cN}_{XY} \omega_{X}^{\frac{1}{2}} 
&\leq \omega_{X}^{\frac{1}{2}} \left( J^{\cM}_{XY} \#_{1/\alpha} A_{XY}  \right) \omega_{X}^{\frac{1}{2}} \\
&\leq (\omega_{X}^{\frac{1}{2}} J^{\cM}_{XY} \omega_{X}^{\frac{1}{2}}) \#_{1/\alpha} A^{\omega}_{XY} \ ,
\end{align*} 
where we used the transformer inequality. As such $A_{XY}^{\omega}$ is feasible for the defining optimization program for $\newQ_{\alpha}(\omega_{X}^{\frac{1}{2}} J^{\cN}_{XY} \omega_{X}^{\frac{1}{2}}  \| \omega_{X}^{\frac{1}{2}} J^{\cM}_{XY} \omega_{X}^{\frac{1}{2}})$, and this implies $\newQ_{\alpha}(\omega_{X}^{\frac{1}{2}} J^{\cN}_{XY} \omega_{X}^{\frac{1}{2}}  \| \omega_{X}^{\frac{1}{2}} J^{\cM}_{XY} \omega_{X}^{\frac{1}{2}}) \leq \lambda$. Taking the supremum over $\omega_{X}$ completes the proof.
\end{proof}

An immediate corollary is that the channel divergence is subadditive.
\begin{corollary}
\label{cor:channel_div_subadditive}
For any $\alpha \in (1,\infty)$ and completely positive maps $\cN_1, \cN_2, \cM_1, \cM_2$, we have
\begin{align*}
\newD_{\alpha}(\cN_1 \otimes \cN_2 \| \cM_1 \otimes \cM_2) &\leq \newD_{\alpha}(\cN_1 \| \cM_1 ) + \newD_{\alpha}( \cN_2 \| \cM_2)
\end{align*}
\end{corollary}
\begin{proof}
Let $A^{1}_{X_1Y_1}$ be a feasible solution for the program~\eqref{eq:expr_channel_Q} for the channels $\cN_1$ and $\cM_1$ and $A^{2}_{X_2Y_2}$ for $\cN_2$ and $\cM_2$. Then using the fact that $J^{\cN_1 \otimes \cN_2} = J^{\cN_1} \otimes J^{\cN_2}$ and the tensor product property of the mean (Proposition~\ref{prop:p}), we have that $A^{12}_{X_1X_2 Y_1 Y_2} = A^{1}_{X_1Y_1} \otimes A^{2}_{X_2Y_2}$ is feasible for $\cN_1 \otimes \cN_2$ and $\cM_1 \otimes \cM_2$ and $\| \tr_{Y_1Y_2} A^{12}_{X_1X_2 Y_1 Y_2}  \|_{\infty} = \| \tr_{Y_1} A^{1}_{X_1Y_1} \|_{\infty} \| \tr_{Y_2} A^{2}_{X_2Y_2} \|_{\infty}$.
\end{proof}

Next, we prove that as for states, the regularized channel divergence \change{is equal to} the regularized sandwiched R\'enyi divergence. 
Note that unlike for states, the sandwiched R\'enyi divergence of channels is not additive in general, see~\cite{FFRS19} for an example.
\begin{lemma}
\label{lem:reg_channel_div}
Let $\cN$ and $\cM$ be completely positive maps from $\Lin(X)$ to $\Lin(Y)$. For any $n \geq 1$ and $\alpha \in (1,\infty)$,
\begin{align*}
\frac{1}{n} \newD_{\alpha}(\cN^{\otimes n} \| \cM^{\otimes n}) - \frac{1}{n} \frac{\alpha}{\alpha-1} (d^2+d) \log (n+d) \leq \frac{1}{n} \widetilde{D}_{\alpha}(\cN^{\otimes n} \| \cM^{\otimes n} ) \leq  \frac{1}{n} \newD_{\alpha}(\cN^{\otimes n} \| \cM^{\otimes n}) \ ,
\end{align*}
where $d = \dim X \dim Y$.
\end{lemma}
\begin{proof}
The second inequality follows immediately from the fact that $\widetilde{D}_{\alpha} \leq \newD_{\alpha}$; see Proposition~\ref{prop:relation_other_div}.

For the other direction, the channels $\cN^{\otimes n}$ and $\cM^{\otimes n}$ are covariant with respect to the representation of the symmetric group~$\mathfrak{S}_n$. In fact, for $\pi \in \mathfrak{S}_n$, if we denote $P_{X}(\pi)$ the operator on the space $X^{\otimes n}$ that permutes the $n$ tensor factors according to $\pi$, then
\begin{align*}
\cN^{\otimes n}(P_{X}(\pi) \ . \ P_{X}(\pi)^*) = P_{Y}(\pi) \cN^{\otimes n}(\  . \ ) P_{Y}(\pi)^{*} \ ,
\end{align*} 
and similarly the same relation holds for $\cM^{\otimes n}$. Using the definition~\eqref{eq:def_channel_div}, we have
\begin{align*}
\newD_{\alpha}(\cN^{\otimes n} \| \cM^{\otimes n}) &= \sup_{\omega_{X^n} \in \D(X^{\otimes n})} \newD_{\alpha}( \omega_{X^n}^{\frac{1}{2}}(J_{XY}^{\cN})^{\otimes n} \omega_{X^n}^{\frac{1}{2}} \| \omega_{X^n}^{\frac{1}{2}}(J^{\cM}_{XY})^{\otimes n} \omega_{X^n}^{\frac{1}{2}} ) \ .
\end{align*}
Using \cite[Proposition II.4]{LKDW18} for the divergence $\newD_{\alpha}$ (which satisfies the data-processing inequality as shown in Proposition~\ref{prop:data_processing}), we may restrict the optimization to permutation-invariant states and get
\begin{align*}
\newD_{\alpha}(\cN^{\otimes n} \| \cM^{\otimes n}) &= \sup_{\substack{\omega_{X^n} \in \D(X^{\otimes n}) \\ [P_{X}(\pi), \omega_{X^n}] = 0 \: \forall \pi \in \mathfrak{S}_n}} \newD_{\alpha}( \omega_{X^n}^{\frac{1}{2}}(J_{XY}^{\cN})^{\otimes n} \omega_{X^n}^{\frac{1}{2}} \| \omega_{X^n}^{\frac{1}{2}}(J^{\cM}_{XY})^{\otimes n} \omega_{X^n}^{\frac{1}{2}} ) \ .
\end{align*}
Now consider such a permutation-invariant $\omega_{X^n}$ and we use the relation to the measured R\'enyi divergence in~\eqref{eq:relation_dmeas}:
\begin{align*}
&\newD_{\alpha}( \omega_{X^n}^{\frac{1}{2}}(J^{\cN}_{XY})^{\otimes n} \omega_{X^n}^{\frac{1}{2}} \| \omega_{X^n}^{\frac{1}{2}}(J^{\cM}_{XY})^{\otimes n} \omega_{X^n}^{\frac{1}{2}} ) \\
&\leq D^{\mathbb{M}}_{\alpha}( \omega_{X^n}^{\frac{1}{2}}(J^{\cN}_{XY})^{\otimes n} \omega_{X^n}^{\frac{1}{2}} \| \omega_{X^n}^{\frac{1}{2}}(J^{\cM}_{XY})^{\otimes n} \omega_{X^n}^{\frac{1}{2}} ) + \frac{\alpha}{\alpha-1} \log |\mathrm{spec}(\omega_{X^n}^{\frac{1}{2}}(J^{\cM}_{XY})^{\otimes n} \omega_{X^n}^{\frac{1}{2}})|  \ .
\end{align*}
Now note that if $\omega_{X^n}$ is permutation-invariant, then so is the operator $\omega_{X^n}^{\frac{1}{2}}(J^{\cM}_{XY})^{\otimes n} \omega_{X^n}^{\frac{1}{2}}$ on $(X \otimes Y)^{\otimes n}$. As such, using Lemma~\ref{lem:perm_inv_spec}, 
\begin{align*}
|\mathrm{spec}(\omega_{X^n}^{\frac{1}{2}}(J^{\cM}_{XY})^{\otimes n} \omega_{X^n}^{\frac{1}{2}})| \leq (n+1)^{d} (n + d)^{d^2} \ ,
\end{align*}
where $d := \dim X \dim Y$. Taking the supremum over all $\omega_{X^n}$, we get
\begin{align*}
\newD_{\alpha}(\cN^{\otimes n} \| \cM^{\otimes n}) &\leq D^{\mathbb{M}}_{\alpha}(\cN^{\otimes n} \| \cM^{\otimes n}) + \frac{\alpha}{\alpha-1} (d^2+d) \log (n+d) \\
&\leq \widetilde{D}_{\alpha}(\cN^{\otimes n} \| \cM^{\otimes n}) + \frac{\alpha}{\alpha-1} (d^2+d) \log (n+d) \ .
\end{align*}
This gives the desired result.
\end{proof}

The channel divergence satisfies a chain rule property for any $\alpha \in (1, \infty)$, as the one satisfied for the geometric divergence $\widehat{D}_{\alpha}$~\cite{FF19}.
\begin{proposition}
\label{prop:chain_rule_im}
Let $\alpha \in (1,\infty)$. Let $\rho_{RX'}, \sigma_{RX'} \in \Pos(R \otimes X')$, $\cN, \cM$ be completely positive maps from $\Lin(X')$ to $\Lin(Y)$. Then
\begin{align*}
\newD_{\alpha}((\cI_{R} \otimes \cN_{X' \to Y})(\rho_{RX'}) \| (\cI_{R} \otimes \cM_{X' \to Y})(\sigma_{RX'} ) &\leq \newD_{\alpha}(\cN \| \cM ) + \newD_{\alpha}( \rho_{RX'} \| \sigma_{RX'}) \ .
\end{align*}
\end{proposition}
\begin{proof}
Let $A_{XY}$ be an optimal solution for \eqref{eq:expr_channel_Q} for the maps $\cN$ and $\cM$ and $\bar{A}_{RX'}$ be an optimal solution for \eqref{eq:expr_min_def} for $\rho_{RX'}$ and $\sigma_{RX'}$. Then note that 
\begin{align*}
(\cI_{R} \otimes \cN_{X' \to Y})(\rho_{RX'}) = \bra{\Phi_{XX'}} J^{\cN}_{XY} \otimes \rho_{RX'} \ket{\Phi_{XX'}} \quad (\cI_{R} \otimes \cM_{X' \to Y})(\sigma_{RX'}) = \bra{\Phi_{XX'}} J^{\cM}_{XY} \otimes \sigma_{RX'} \ket{\Phi_{XX'}} \ .
\end{align*}
Combining the properties of $A_{XY}$ and $\bar{A}_{RX'}$ using Proposition~\ref{prop:p}, we get
\begin{align*}
J^{\cN}_{XY} \otimes \rho_{RX'} \leq (J^{\cM}_{XY} \otimes \sigma_{RX'}) \#_{1/\alpha} (A_{XY} \otimes \bar{A}_{RX'}) \ .
\end{align*}
Then using the transformer inequality, we have
\begin{align*}
\bra{\Phi_{XX'}} J^{\cN}_{XY} \otimes \rho_{RX'}  \ket{\Phi_{XX'}}
&\leq \bra{\Phi_{XX'}} \big(  (J^{\cM}_{XY} \otimes \sigma_{RX'}) \#_{1/\alpha} (A_{XY} \otimes \bar{A}_{RX'}) \big) \ket{\Phi_{XX'}} \\
&\leq \big( \bra{\Phi_{XX'}} (J^{\cM}_{XY} \otimes \sigma_{RX'})  \ket{\Phi_{XX'}} \big) \#_{1/\alpha} \big(\bra{\Phi_{XX'}} (A_{XY} \otimes \bar{A}_{RX'}) \ket{\Phi_{XX'}}\big) \ .
\end{align*}
To conclude it suffices to compute 
\begin{align*}
\tr\big(\bra{\Phi_{XX'}} (A_{XY} \otimes \bar{A}_{RX'}) \ket{\Phi_{XX'}}\big) &= \bra{\Phi_{XX'}} \tr_{Y}(A_{XY}) \otimes \tr_{R}(\bar{A}_{RX'}) \ket{\Phi_{XX'}} \\
&\leq \| \tr_{Y} A_{XY} \|_{\infty} \bra{\Phi_{XX'}} \id_{X} \otimes \tr_{R}(\bar{A}_{RX'}) \ket{\Phi_{XX'}} \\
&= \| \tr_{Y} A_{XY} \|_{\infty} \tr(\bar{A}_{RX'}) \ ,
\end{align*}
which after taking the logarithm establishes the desired inequality.
\end{proof}

\begin{remark}
Note that the chain rule can be seen as a generalization of the data processing inequality. In fact, we can take the $R$ system to be trivial and if the maps are the same $\cN = \cM$ and in addition trace-preserving, then $\newD_{\alpha}(\cN \| \cM) = 0$.
\end{remark}

\section{Applications}

\label{sec:applications}

In this section we present some example applications of the newly introduced divergences. Most of these applications are related to the regularized sandwiched divergence between channels. For $\alpha > 1$, we denote
%
\begin{align}
\label{eq:def_reg_div}
\widetilde{D}_{\alpha}^{\reg}(\cN \| \cM) := \lim_{n \to \infty} \frac{1}{n} \widetilde{D}_{\alpha}(\cN^{\otimes n} \| \cM^{\otimes n}) \ .
\end{align}
We note that as the sequence $\frac{1}{n} \widetilde{D}_{\alpha}(\cN^{\otimes n} \| \cM^{\otimes n})$ is superadditive, using Fekete's lemma, the limit exists and can be replaced by a supremum over $n$. Regularized entropic quantities appear extensively in quantum information theory but it is unclear how to compute them (or even whether they are computable to start with) as we do not have control on the convergence speed in the regularization. Using $\newD_{\alpha}$, one can quantify the convergence speed explicitly for $\widetilde{D}^{\reg}_{\alpha}$ and thus show that this quantity is computable. 

\subsection{Converging hierarchy of upper bounds on the regularized divergence of channels}
\label{sec:hierarchy}

\begin{theorem}
\label{thm:comp_reg_sand}
Let $\alpha \in (1, \infty)$ and $\cN, \cM$ be completely positive maps from $\Lin(X)$ to $\Lin(Y)$. Then for any $m \geq 1$,
\begin{align*}
\frac{1}{m} \newD_{\alpha}(\cN^{\otimes m} \| \cM^{\otimes m}) - \frac{1}{m} \frac{\alpha}{\alpha-1} (d^2+d) \log (m+d) \leq \: \widetilde{D}_{\alpha}^{\reg}(\cN \| \cM) \: \leq \frac{1}{m} \newD_{\alpha}(\cN^{\otimes m} \| \cM^{\otimes m}) \ ,
\end{align*}
where $d = \dim X \dim Y$. We can also write 
\begin{align}
\label{eq:convergence_speed_sand_channel}
 \widetilde{D}_{\alpha}^{\reg}(\cN \| \cM) - \frac{1}{m} \widetilde{D}_{\alpha}(\cN^{\otimes m} \| \cM^{\otimes m})  \leq \frac{1}{m} \frac{\alpha}{\alpha-1} (d^2+d) \log (m+d) \ .
\end{align}
\end{theorem}
\begin{proof}
The lower bound follows immediately from Lemma~\ref{lem:reg_channel_div} with $m = n$. For the upper bound, using the fact that $\widetilde{D}_{\alpha} \leq \newD_{\alpha}$ and the subadditivity property in Corollary~\ref{cor:channel_div_subadditive} we have for any $n, m$
\begin{align*}
\widetilde{D}_{\alpha}( \cN^{\otimes mn} \| \cM^{\otimes mn} ) 
&\leq \newD_{\alpha}( \cN^{\otimes mn} \| \cM^{\otimes mn} ) \\
&\leq n \newD_{\alpha}( \cN^{\otimes m} \| \cM^{\otimes m} ) \ .
\end{align*}
Dividing by $mn$ and taking the limit as $n \to \infty$ concludes the proof of the upper bound.
The inequality~\eqref{eq:convergence_speed_sand_channel} follows from this upper bound together with Lemma~\ref{lem:reg_channel_div} (more specifically, the lower bound there applied for $m = n$).
\end{proof}
Note that for any finite $m$, the quantity $\frac{1}{m} \newD_{\alpha}(\cN^{\otimes m} \| \cM^{\otimes m})$ can be approximated to arbitrary accuracy and this shows that $\widetilde{D}_{\alpha}^{\reg}$ can be approximated within additive $\eps$ in finite time. 
The precise analysis of the running time as a function of the bit size of the input is a subtle question that is outside the scope of this work. 
But staying at a high level, the running time of this algorithm will be exponential in the input and output dimensions of the channels. In fact, we can take $m = \lceil \frac{8\alpha d^3}{(\alpha-1)\eps} \rceil$ where $d$ is the dimension of the Choi state of the channels $\cN$ and $\cM$, and then compute $\newD_{\alpha}(\cN^{\otimes m} \| \cM^{\otimes m})$. The channels $\cN^{\otimes m}$ and $\cM^{\otimes m}$ have a Choi state of dimension $d^m$ and thus the convex program defining $\newD_{\alpha}(\cN^{\otimes m} \| \cM^{\otimes m})$ can be approximated in time that is polynomial in $d^m$ using the ellipsoid algorithm. Overall, the running time is exponential in $d$. 

As the regularized divergence between channels appears in the analysis of many information processing tasks, we believe this result will be useful in obtaining improved characterizations of such tasks. An example is the task of channel discrimination, for which the regularized Umegaki channel divergence governs the asymptotic error rate~\cite{berta18,Wang2019,FFRS19}. One could also obtain upper bounds on quantum channel capacities, such as the classical capacity, in terms of regularized divergence between channels (see e.g.,~\cite{WFT17}). In fact, closely following the approach of~\cite{FF19} to upper bound the classical capacity of a quantum channel and replacing $\widehat{D}_{\alpha}$ with $\newD_{\alpha}$, one does obtain improved bounds, including for the amplitude damping channel. However, the improvements obtained by such a direct application of~\cite{FF19} were small, typically less than 1\%. To give an example, for a damping parameter $\gamma = 0.5$, we obtain (using $\newD_{\alpha}$ with two copies of the channel) an upper bound on the capacity of $0.7694...$ whereas the previous bound (using $D_{\max}$ or $\widehat{D}_{\alpha}$) was $0.7716...$~\cite{WXD17,WFT17}. We leave the further exploration of this question for future work.

\subsection{A chain rule for the sandwiched R\'enyi divergence}

\label{sec:chain_rule_sand}

Our second application is a chain rule for the sandwiched R\'enyi divergence, which once again features the regularized divergence between channels. Such a chain rule was proved in~\cite{FFRS19} for the Umegaki relative entropy.
\begin{corollary}[Chain rule for the sandwiched R\'enyi divergence]
\label{cor:chainrulesandwiched}
Let $\alpha \in (1,\infty)$. Let $\rho_{RX'}, \sigma_{RX'} \in \Pos(RX')$, $\cN, \cM$ be completely positive maps from $\Lin(X')$ to $\Lin(Y)$. Then
\begin{align}
\label{eq:chain_rule_sand}
\widetilde{D}_{\alpha}((\cI_{R} \otimes \cN_{X' \to Y})(\rho_{RX'}) \| (\cI_{R} \otimes \cM_{X' \to Y})(\sigma_{RX'} ) ) &\leq  \widetilde{D}^{\reg}_{\alpha}(\cN \| \cM ) + \widetilde{D}_{\alpha}( \rho_{RX'} \| \sigma_{RX'}) \ .
\end{align}
\end{corollary}
\begin{proof}
We apply the chain rule in Proposition~\ref{prop:chain_rule_im} to the states $\rho_{RX'}^{\otimes n}$ and $\sigma_{RX'}^{\otimes n}$ and the channels $\cN^{\otimes n}$ and $\cM^{\otimes n}$ and get
\begin{align*}
\frac{1}{n} \newD_{\alpha}((\cI_{R} \otimes \cN_{X' \to Y})(\rho_{RX'})^{\otimes n} \| (\cI_{R} \otimes \cM_{X' \to Y})(\sigma_{RX'} )^{\otimes n}) &\leq  \frac{1}{n} \newD_{\alpha}(\cN^{\otimes n} \| \cM^{\otimes n} ) + \frac{1}{n} \newD_{\alpha}( \rho_{RX'}^{\otimes n} \| \sigma_{RX'}^{\otimes n}) \ .
\end{align*}
Taking the limit as $n \to \infty$, the state divergences becomes sandwiched divergences using~\eqref{eq:reg_im_sand} and the channel divergence becomes the regularized channel sandwiched divergence using Lemma~\ref{lem:reg_channel_div}.
\end{proof}
\begin{remark}
It is unclear whether taking the limit $\alpha \to 1$ in this chain rule recovers the chain rule proved in~\cite{FFRS19}. The reason for this difficulty is that it remains open whether $\lim_{\alpha \downarrow 1} \widetilde{D}^{\reg}_{\alpha}(\cN \| \cM ) = D^{\reg}(\cN \| \cM )$.
\end{remark}

It is also possible to phrase the chain rule in terms of amortized divergences as introduced in~\cite{berta18}. For a divergence $\mathbf{D}$, the amortized divergence is defined as
\begin{align}
\label{eq:def_amor_div}
\mathbf{D}^{\amor}(\cN \| \cM) = \sup_{\rho_{RX} \in \D(RX), \sigma_{RX} \in \D(RX)} \mathbf{D}((\cI_{R} \otimes \cN)(\rho_{RX}) \| (\cI_{R} \otimes \cM)(\sigma_{RX}) ) - \mathbf{D}(\rho_{RX} \| \sigma_{RX}) \ ,
\end{align}
where the supremum \change{also runs} over all finite dimensional spaces $R$. When $\mathbf{D}$ is the sandwiched R\'enyi divergence $\widetilde{D}_{\alpha}$, note that for positive real numbers $\beta$ and $\gamma$, we have $\widetilde{D}_{\alpha}(\beta \rho \| \gamma \sigma ) = \widetilde{D}_{\alpha}(\rho \| \sigma) + \frac{\alpha}{\alpha - 1} \log \beta - \log \gamma$. As a result, for any nonzero $\rho_{RX}, \sigma_{RX} \in \Pos(RX)$, we have
\begin{align}
&\widetilde{D}_{\alpha}((\cI_{R} \otimes \cN)(\rho_{RX}) \| (\cI_{R} \otimes \cM)(\sigma_{RX}) ) - \widetilde{D}_{\alpha}(\rho_{RX} \| \sigma_{RX}) \notag \\
&= \widetilde{D}_{\alpha}\left((\cI_{R} \otimes \cN)\left(\frac{\rho_{RX}}{\tr(\rho)}\right) \| (\cI_{R} \otimes \cM)\left(\frac{\sigma_{RX}}{\tr(\sigma)}\right) \right) - \widetilde{D}_{\alpha}\left(\frac{\rho_{RX}}{\tr(\rho)} \| \frac{\sigma_{RX}}{\tr(\sigma)}\right) \label{eq:amor_unnorm} \ ,
\end{align}
which means that in~\eqref{eq:def_amor_div}, we can also take the supremum over all nonzero positive semidefinite operators.
\begin{theorem}[Amortization = regularization for sandwiched divergence]
For any completely positive maps $\cN, \cM$ and any $\alpha > 1$, we have
\begin{align}
\label{eq:amor=reg}
\widetilde{D}_{\alpha}^{\amor}(\cN \| \cM) = \widetilde{D}_{\alpha}^{\reg}(\cN \| \cM) \ .
\end{align}
\end{theorem}
\begin{proof}
The inequality $\leq$ follows immediately from the chain rule in~\eqref{eq:chain_rule_sand}.

The inequality $\geq$ is actually true for any generalized divergence and was observed in previous works~\cite{berta18,Wang2019}. Note that we can equivalently write the channel divergence as 
\begin{align*}
\widetilde{D}_{\alpha}(\cN \| \cM) = \sup_{\phi_{XX'} \in \D(X \otimes X')} \widetilde{D}_{\alpha}((\cI_{X} \otimes \cN_{X' \to Y})(\phi) \| (\cI_{X} \otimes \cM_{X' \to Y})(\phi))  \ ,
\end{align*}
where as usual $X'$ and $X$ have the same dimension. Thus, denoting the $n$ copies of $X$ by $X_1, X_2, \dots, X_n$ and using the shorthand $X_{i}^{j}$ to denote $X_i X_{i+1} \dots X_j$, we have
\begin{align*}
\widetilde{D}_{\alpha}(\cN^{\otimes n} \| \cM^{\otimes n})
&= \sup_{\phi_{X_1^n {X'}^n_1} \in \D(X^{\otimes n} \otimes {X'}^{\otimes n})} \widetilde{D}_{\alpha}((\cI_{X_1^n} \otimes \cN_{X' \to Y}^{\otimes n})(\phi) \| (\cI_{X_1^n} \otimes \cM_{X' \to Y}^{\otimes n})(\phi)) \\
&= \sup_{\phi_{X_1^n {X'}^n_1}}  \sum_{i=0}^{n-1} \Bigg( \widetilde{D}_{\alpha}((\cI_{X_1^n} \otimes \cN_{X' \to Y}^{\otimes (i+1)} \otimes \cI_{{X'}_{i+2}^n})(\phi) \| (\cI_{X_1^n {X'}_{i+2}^n} \otimes \cM_{X' \to Y}^{\otimes (i+1)} \otimes \cI_{{X'}_{i+2}^n})(\phi)) \\
&\qquad \qquad - \widetilde{D}_{\alpha}((\cI_{X_1^n} \otimes \cN_{X' \to Y}^{\otimes i} \otimes \cI_{{X'}_{i+1}^n})(\phi) \| (\cI_{X_1^n {X'}_{i+1}^n} \otimes \cM_{X' \to Y}^{\otimes i} \otimes \cI_{{X'}_{i+1}^n})(\phi)) \Bigg) \\
&\leq \sum_{i=0}^{n-1} \sup_{\phi_{X_1^n {X'}^n_1}}  \Bigg( \widetilde{D}_{\alpha}((\cI_{X_1^n} \otimes \cN_{X' \to Y}^{\otimes (i+1)} \otimes \cI_{{X'}_{i+2}^n})(\phi) \| (\cI_{X_1^n {X'}_{i+2}^n} \otimes \cM_{X' \to Y}^{\otimes (i+1)} \otimes \cI_{{X'}_{i+2}^n})(\phi)) \\
&\qquad \qquad - \widetilde{D}_{\alpha}((\cI_{X_1^n} \otimes \cN_{X' \to Y}^{\otimes i} \otimes \cI_{{X'}_{i+1}^n})(\phi) \| (\cI_{X_1^n {X'}_{i+1}^n} \otimes \cM_{X' \to Y}^{\otimes i} \otimes \cI_{{X'}_{i+1}^n})(\phi)) \Bigg) \ .
\end{align*}
Note that in the $i$-th term, we subtract two expressions that differ by an application of the channels $\cN$ and $\cM$ on the system $X_{i+1}$ and so the remaining systems can be considered as the $R$ system in the definition~\eqref{eq:def_amor_div}. Using in addition the observation in~\eqref{eq:amor_unnorm} saying that $\rho$ and $\sigma$ need not be normalized, we get that each term is bounded by $\widetilde{D}_{\alpha}^{\amor}(\cN \| \cM)$. Thus,
\begin{align*}
\widetilde{D}_{\alpha}(\cN^{\otimes n} \| \cM^{\otimes n})
&\leq n \widetilde{D}_{\alpha}^{\amor}(\cN \| \cM) \ ,
\end{align*}
which gives the desired result.
\end{proof}

The concept of amortization is particularly useful when analyzing information processing tasks that have an adaptive aspect. We discuss some examples below.

\subsubsection{Channel discrimination}
\label{sec:channel_discrimination}

We discuss the example of channel discrimination, referring to~\cite{Cooney2016,berta18} for a more detailed and precise presentation of the problem and the relevant references on the topic. Imagine we would like to distinguish between two quantum channels $\cN$ and $\cM$ having black box access to $n$ uses of one of them. The task of adaptive channel discrimination is to decide which channel we are dealing with. The word adaptive here refers to the fact that our use of one of the black boxes can depend on the outcomes of a previously used black box. By contrast, a strategy is called parallel (or nonadaptive) if the $n$ black boxes are used in parallel on a fixed input state. As is common in hypothesis testing, we call the type I error the probability $\alpha_n$ that the channel is actually $\cN$ but our procedure says $\cM$ and the type II error $\beta_n$ is the other kind of error and the goal is to determine the tradeoff between these two errors. Multiple regimes can be considered, the most studied is the asymmetric or Stein setting where we set $\alpha_n \leq \eps$ for some $\eps \in (0,1)$ and consider the asymptotic behavior of the optimal type II error $-\frac{1}{n} \log \beta_n$. The works~\cite{berta18,Wang2019,FFRS19} establish that if we take $\eps \to 0$ this is given by the regularized Umegaki relative entropy $D^{\reg}(\cN \| \cM)$\footnote{We recall that the Umegaki quantum relative entropy (also called simply quantum relative entropy) is defined as $D(\rho \| \sigma) = \tr(\rho (\log \rho - \log \sigma))$~\cite{umegaki1962conditional,HP91}}. Our focus here is on the strong converse regime, i.e., we require $\beta_n \leq 2^{-rn}$ with $r > D^{\reg}(\cN \| \cM)$ and we consider the behavior of $\alpha_n$. As far as we are aware, it is not known whether in this case we always have $\alpha_n \to 1$ (this would be a strong converse property). However, we can always consider the following quantity which measures how quickly $\alpha_n$ goes to $1$ when it does so:
\begin{align*}
\overline{H}(r, \cN, \cM) := \inf_{\text{adaptive strategies}} \Big\{ \limsup_{n \to \infty}  -\frac{1}{n} \log(1 - \alpha_n) : \change{\limsup_{n \to \infty} \frac{1}{n} \log \beta_n \leq -r} \Big\} \ .
\end{align*}
Note that if $\alpha_n$ does not converge to $1$ exponentially fast, then this quantity is $0$.

A lower bound is given in~\cite[Proposition 20]{berta18} for this quantity:
\begin{align}
\label{eq:error_exponent_strong_converse_amor}
\overline{H}(r, \cN, \cM)
&\geq \sup_{\alpha > 1} \frac{\alpha - 1}{\alpha} \left( r -  \widetilde{D}_{\alpha}^{\amor}(\cN \| \cM) \right) \ .
\end{align}
\change{
We note that~\cite[Proposition 20]{berta18} actually shows something stronger: for any $n$, any adaptive strategy and any $\alpha > 1$, we have
\begin{align*}
-\frac{1}{n} \log(1 - \alpha_n) \geq \frac{\alpha - 1}{\alpha}\left(-\frac{1}{n} \log \beta_n - \widetilde{D}_{\alpha}^{\amor}(\cN \| \cM)\right) \ . 
\end{align*}
Thus, for any family of strategies satisfying $\limsup_{n \to \infty} \frac{1}{n} \log \beta_n \leq -r$, we have
\begin{align*}
\limsup_{n \to \infty} -\frac{1}{n} \log(1 - \alpha_n) &
\geq \liminf_{n \to \infty} \frac{\alpha - 1}{\alpha}\left(-\frac{1}{n} \log \beta_n - \widetilde{D}_{\alpha}^{\amor}(\cN \| \cM)\right) \\
&=  \frac{\alpha - 1}{\alpha}\left(- \limsup_{n \to \infty}\frac{1}{n} \log \beta_n - \widetilde{D}_{\alpha}^{\amor}(\cN \| \cM)\right) \\
&\geq \frac{\alpha - 1}{\alpha}\left(r - \widetilde{D}_{\alpha}^{\amor}(\cN \| \cM)\right) \ ,
\end{align*}
which establishes~\eqref{eq:error_exponent_strong_converse_amor}.}

Using equality~\eqref{eq:amor=reg} together with the explicit convergence bounds in Theorem~\ref{thm:comp_reg_sand} as well as the strong converse exponent established for states~\cite{mosonyi2015quantum}, we show that this bound is in fact tight. This generalizes the result of~\cite{Cooney2016} who considered the case where $\cM$ is a replacer channel, i.e., $\cM(W) = \tr(W) \sigma$ for some state $\sigma$.
\begin{theorem}
For any completely positive and trace-preserving maps $\cN, \cM$ and any $r > 0$, we have
\begin{align}
\overline{H}(r, \cN, \cM)
&= \sup_{\alpha > 1} \frac{\alpha - 1}{\alpha} \left( r -  \widetilde{D}_{\alpha}^{\reg}(\cN \| \cM) \right) \ .
\end{align}
In addition, the achievability uses a nonadaptive strategy and this shows that adaptive strategies do not offer an advantage in this setting.
\end{theorem}
\begin{remark}[Continuity of $\widetilde{D}_{\alpha}^{\reg}$ when $\alpha \to 1$]
Note that this result implies that $\overline{H}(r, \cN, \cM) = 0$ if $r \leq \inf_{\alpha > 1} \widetilde{D}_{\alpha}^{\reg}(\cN \| \cM)$ and $\overline{H}(r, \cN, \cM) > 0$ if $r > \inf_{\alpha > 1} \widetilde{D}^{\reg}_{\alpha}(\cN \| \cM)$. As the behaviour of $\widetilde{D}_{\alpha}^{\reg}$ as $\alpha \to 1$ remains unclear, we cannot rule out that $D^{\reg}(\cN \| \cM) < \inf_{\alpha > 1} \widetilde{D}_{\alpha}^{\reg}(\cN \| \cM)$ for some channels, and so it remains open whether a strong converse property holds in general.
\end{remark}
\begin{proof}
As usual, we will assume $J^{\cN} \ll J^{\cM}$, as otherwise, $D^{\reg}(\cN \| \cM) = \infty$ and the statement is void. The lower bound $\geq$ follows immediately from~\eqref{eq:error_exponent_strong_converse_amor} and equality~\eqref{eq:amor=reg}.

For the upper bound, the idea is to use the characterization of~\cite{mosonyi2015quantum} for the strong converse exponent for state discrimination. They show that for any states $\rho$ and $\sigma$ and $r > 0$, there is a family of strategies to distinguish between $\rho^{\otimes n}$ and $\sigma^{\otimes n}$ with type II error probability $\beta_n(\rho, \sigma) \leq 2^{-rn}$ and achieving a type I error probability $\alpha_n(\rho, \sigma)$ satisfying
\begin{align*}
\limsup_{n \to \infty} -\frac{1}{n}\log(1-\alpha_n(\rho, \sigma)) \leq \sup_{\alpha > 1} \frac{\alpha - 1}{\alpha} \left( r -  \widetilde{D}_{\alpha}(\rho \| \sigma) \right) \ .
\end{align*}

Let $\eps > 0$ and choose an integer $m$ so that $\frac{1}{m} (d^2+d) \log (m+d) < \eps$ where $d = \dim X \dim Y$. For any state $\omega \in \D(X^{\otimes m})$ we can apply this result to the states $\omega^{\frac{1}{2}} J^{\cN^{\otimes m}} \omega^{\frac{1}{2}}$ and $\omega^{\frac{1}{2}} J^{\cM^{\otimes m}} \omega^{\frac{1}{2}}$ and get a sequence of strategies achieving $\limsup_{n \to \infty} \frac{1}{n} \log \beta_{n}(\omega^{\frac{1}{2}} J^{\cN^{\otimes m}} \omega^{\frac{1}{2}}, \omega^{\frac{1}{2}} J^{\cM^{\otimes m}} \omega^{\frac{1}{2}}) \leq -rm$ and
\begin{align}
\label{eq:strategy_states_blocks}
\limsup_{n \to \infty} -\frac{1}{n}\log(1-\alpha_n(\omega^{\frac{1}{2}} J^{\cN^{\otimes m}} \omega^{\frac{1}{2}}, \omega^{\frac{1}{2}} J^{\cM^{\otimes m}} \omega^{\frac{1}{2}})) \leq \sup_{\alpha > 1} \frac{\alpha - 1}{\alpha} \left( rm -  \widetilde{D}_{\alpha}(\omega^{\frac{1}{2}} J^{\cN^{\otimes m}} \omega^{\frac{1}{2}} \| \omega^{\frac{1}{2}} J^{\cM^{\otimes m}} \omega^{\frac{1}{2}}) \right) \ .
\end{align}
We choose $\omega_m$ to achieve up to $\eps$ the infimum over $\omega \in \D(X^{\otimes m})$ with $\omega > 0$ of the right hand side and for this $\omega_m$, we have a strategy achieving 
\begin{align*}
\limsup_{n \to \infty} \frac{1}{n} \log \beta_{n}(\omega_m^{\frac{1}{2}} J^{\cN^{\otimes m}} \omega_m^{\frac{1}{2}}, \omega_m^{\frac{1}{2}} J^{\cM^{\otimes m}} \omega_m^{\frac{1}{2}}) \leq 2^{-rm \cdot n}
\end{align*}
and
\begin{align}
\notag
&\limsup_{n \to \infty} -\frac{1}{n}\log(1-\alpha_n(\omega_m^{\frac{1}{2}} J^{\cN^{\otimes m}} \omega_m^{\frac{1}{2}}, \omega_m^{\frac{1}{2}} J^{\cM^{\otimes m}} \omega_m^{\frac{1}{2}})) \\
&\leq \inf_{\substack{\omega \in \D(X^{\otimes m}) \\ \omega > 0}} \sup_{\alpha > 1} \frac{\alpha - 1}{\alpha} \left( rm -  \widetilde{D}_{\alpha}(\omega^{\frac{1}{2}} J^{\cN^{\otimes m}} \omega^{\frac{1}{2}} \| \omega^{\frac{1}{2}} J^{\cM^{\otimes m}} \omega^{\frac{1}{2}}) \right) + \eps \ .
\label{eq:strategy_states_blocks_inf_omega}
\end{align}

Now we observe that for any channels $\cA$ and $\cB$ such that $J^{\cA} \ll J^{\cB}$, we can perform the change of variable $u = \frac{\alpha - 1}{\alpha}$ and get
\begin{align}
\label{eq:inf_omega_sup_alpha}
\inf_{\substack{\omega \in \D(X) \\ \omega > 0}} \sup_{\alpha > 1} \frac{\alpha - 1}{\alpha} \left( r -  \widetilde{D}_{\alpha}( \omega^{\frac{1}{2}} J^{\cA} \omega^{\frac{1}{2}} \| \omega^{\frac{1}{2}} J^{\cB} \omega^{\frac{1}{2}}) \right)
&= \inf_{\substack{\omega \in \D(X) \\ \omega > 0}} \sup_{u \in (0,1)} f(\omega, u) \ ,
\end{align}
where we defined the function $f : \D(X) \times (0,1) \to \RR$ by $f(\omega, u) = u r -  u \widetilde{D}_{\frac{1}{1-u}}( \omega^{\frac{1}{2}} J^{\cA} \omega^{\frac{1}{2}} \| \omega^{\frac{1}{2}} J^{\cB} \omega^{\frac{1}{2}})$. We extend the function $f$ to $f(\omega,0) = 0$ and $f(\omega,1) = r - D_{\max}( \omega^{\frac{1}{2}} J^{\cA} \omega^{\frac{1}{2}} \| \omega^{\frac{1}{2}} J^{\cB} \omega^{\frac{1}{2}})$. Note that as we assumed $J^{\cA} \ll J^{\cB}$, we have $D_{\max}( \omega^{\frac{1}{2}} J^{\cA} \omega^{\frac{1}{2}} \| \omega^{\frac{1}{2}} J^{\cB} \omega^{\frac{1}{2}}) < \infty$ and we even have for any $\omega > 0$, $D_{\max}( \omega^{\frac{1}{2}} J^{\cA} \omega^{\frac{1}{2}} \| \omega^{\frac{1}{2}} J^{\cB} \omega^{\frac{1}{2}}) = D_{\max}(\cA \| \cB)$ is independent of $\omega$. As we will see shortly, for any $\omega > 0$, the function $u \mapsto f(\omega, u)$ is thus continuous on $[0,1]$. As such we have
\begin{align}
\label{eq:inf_max}
\inf_{\substack{\omega \in \D(X) \\ \omega > 0}} \sup_{u \in (0,1)} f(\omega, u) 
&= 
\inf_{\substack{\omega \in \D(X) \\ \omega > 0}} \max_{u \in [0,1]} f(\omega, u) \ .
\end{align}
We are now ready to apply Sion's minimax theorem. To do this, we check the following conditions:
\begin{itemize}
\item For any $\omega \in \D(X)$ with $\omega > 0$, the function $u \mapsto f(\omega, u)$ is concave and continuous on the compact interval $[0,1]$. This follows from~\cite[Remark IV.13 or the discussion preceding Lemma IV.9]{mosonyi2015quantum} which shows that $u \mapsto u \widetilde{D}_{\frac{1}{1-u}}( \omega^{\frac{1}{2}} J^{\cA} \omega^{\frac{1}{2}} \| \omega^{\frac{1}{2}} J^{\cB} \omega^{\frac{1}{2}})$ is convex and continuous on $[0,1)$. It is also clear that extending it continuously to $u = 1$ preserves the two properties.
\item For any $u \in [0,1]$, the function $\omega \mapsto f(\omega, u)$ is convex and continuous on the convex set $\{\omega \in \cD(X) : \omega > 0\}$. For $u \in \{0,1\}$, this is trivial as the function is constant. For $u \in (0,1)$, this follows immediately from Lemma~\ref{lem:concave_omega}.
\end{itemize}
Applying Sion's minimax theorem, we can exchange the $\inf$ and $\max$ in~\eqref{eq:inf_max} and get
\begin{align*}
\inf_{\substack{\omega \in \D(X) \\ \omega > 0}} \sup_{\alpha > 1} \frac{\alpha - 1}{\alpha} \left( r -  \widetilde{D}_{\alpha}( \omega^{\frac{1}{2}} J^{\cA} \omega^{\frac{1}{2}} \| \omega^{\frac{1}{2}} J^{\cB} \omega^{\frac{1}{2}}) \right)
&= \max_{u \in [0,1]} \inf_{\substack{\omega \in \D(X) \\ \omega > 0}}  u r -  u \widetilde{D}_{\frac{1}{1-u}}( \omega^{\frac{1}{2}} J^{\cA} \omega^{\frac{1}{2}} \| \omega^{\frac{1}{2}} J^{\cB} \omega^{\frac{1}{2}}) \\
&= \max \Big\{\sup_{u \in (0,1)}  u r -  u \widetilde{D}_{\frac{1}{1-u}}( \cA \| \cB ), 0, r - D_{\max}(\cA \| \cB) \Big\} \ .
\end{align*}
Note that for the second equality we used equality~\eqref{eq:sup_channel_div_strict} saying that we can drop the $\omega > 0$ condition in the infimum. Now as $\widetilde{D}_{\alpha} \leq D_{\max}$, we have $\sup_{u \in (0,1)}  u r -  u \widetilde{D}_{\frac{1}{1-u}}( \cA \| \cB ) \geq \sup_{u \in (0,1)} u (r - D_{\max}(\cA \| \cB)) \geq \max\{ 0, r - D_{\max}(\cA \| \cB) \}$ so we can drop the terms $0$ and $r - D_{\max}(\cA \| \cB)$ from the maximization.
Thus, \eqref{eq:strategy_states_blocks_inf_omega} becomes 
\begin{align}
\notag
\limsup_{n \to \infty} -\frac{1}{n}\log(1-\alpha_n(\omega_m^{\frac{1}{2}} J^{\cN^{\otimes m}} \omega_m^{\frac{1}{2}}, \omega_m^{\frac{1}{2}} J^{\cM^{\otimes m}} \omega_m^{\frac{1}{2}})) &\leq \sup_{\alpha > 1} \frac{\alpha - 1}{\alpha} \left( rm -  \widetilde{D}_{\alpha}(\cN^{\otimes m} \| \cM^{\otimes m}) \right) + \eps \ .
\end{align}
Using the finite convergence bounds in~\eqref{eq:convergence_speed_sand_channel} for $\widetilde{D}^{\reg}_{\alpha}$ , we get
\begin{align*}
\frac{\alpha-1}{\alpha} \widetilde{D}^{\reg}_{\alpha}(\cN \| \cM) 
&\leq \frac{\alpha-1}{\alpha} \frac{1}{m} \widetilde{D}_{\alpha}(\cN^{\otimes m} \| \cM^{\otimes m}) +  \frac{1}{m} (d^2+d) \log (m+d) \\
&\leq \frac{\alpha-1}{\alpha} \frac{1}{m} \widetilde{D}_{\alpha}(\cN^{\otimes m} \| \cM^{\otimes m}) +  \eps
\end{align*}
recalling our choice of $m$. Thus,
\begin{align}
\label{eq:approx_reg_m_unif}
 \sup_{\alpha > 1} \frac{\alpha - 1}{\alpha} \left( r -  \frac{1}{m} \widetilde{D}_{\alpha}( \cN^{\otimes m} \| \cM^{\otimes m}) \right) \leq \sup_{\alpha > 1} \frac{\alpha - 1}{\alpha} \left( r -  \widetilde{D}^{\reg}_{\alpha}( \cN \| \cM) \right) + \eps  \ .
\end{align}
As a result, we have
\begin{align}
\notag
\frac{1}{m} \limsup_{n \to \infty} -\frac{1}{n}\log(1-\alpha_n(\omega_m^{\frac{1}{2}} J^{\cN^{\otimes m}} \omega_m^{\frac{1}{2}}, \omega_m^{\frac{1}{2}} J^{\cM^{\otimes m}} \omega_m^{\frac{1}{2}})) &\leq \sup_{\alpha > 1} \frac{\alpha - 1}{\alpha} \left( r -  \widetilde{D}^{\reg}_{\alpha}( \cN \| \cM) \right) + 2 \eps \ .
\end{align}
In other words, we have constructed a sequence of strategies for distinguishing between $\cN^{\otimes mn}$ and $\cM^{\otimes mn}$ for $n \geq 1$ with a type II error $\beta_{nm}$ and a type I error $\alpha_{nm}$ satisfying 
\begin{align*}
\frac{1}{m} \limsup_{n \to \infty} \frac{1}{n}\log \beta_{nm} 
&\leq - r \ .
\end{align*}
and
\begin{align}
\label{eq:alpha_err_subseq}
\frac{1}{m} \limsup_{n \to \infty} -\frac{1}{n}\log(1-\alpha_{nm})
&\leq \sup_{\alpha > 1} \frac{\alpha - 1}{\alpha} \left( r -  \widetilde{D}^{\reg}_{\alpha}( \cN \| \cM) \right) + 2 \eps \ .
\end{align}
\change{
To conclude, we define a strategy for distinguishing between $\cN^{\otimes k}$ and $\cM^{\otimes k}$ for $k$ that is not necessarily of the form $mn$ for some $n$. For that, we write $k = mq + p$ with $0 \leq p < m$ and we only use $mq$ copies of the channel and apply the above argument. We thus obtain exactly the same type I and type II errors as for $mq$ copies, i.e., $\alpha_{k} = \alpha_{mq}$ and $\beta_{k} = \beta_{mq}$. With this notation, we have
\begin{align*}
\frac{1}{k} \log \beta_k = \frac{1}{mq+p} \log \beta_{mq} \leq \frac{1}{m(q+1)} \log \beta_{mq} = \frac{q}{q+1} \frac{1}{m q} \log \beta_{mq} \ .
\end{align*}
As $k \to \infty$, we have $q \to \infty$ so
\begin{align*}
\limsup_{k \to \infty} \frac{1}{k} \log \beta_k \leq \limsup_{q \to \infty} \frac{1}{m q} \log \beta_{mq} \leq -r \ .
\end{align*}
In addition, using the same notation, the type I error satisfies
\begin{align*}
- \frac{1}{k} \log (1-\alpha_k) \leq - \frac{1}{mq} \log (1-\alpha_{mq}) \ .
\end{align*}
As a result, \eqref{eq:alpha_err_subseq} implies that
\begin{align*}
 \limsup_{k \to \infty} -\frac{1}{k}\log(1-\alpha_{k})
&\leq \sup_{\alpha > 1} \frac{\alpha - 1}{\alpha} \left( r -  \widetilde{D}^{\reg}_{\alpha}( \cN \| \cM) \right) + 2 \eps \ .
\end{align*}
}
As this is valid for any $\eps > 0$, we obtain the claimed result. 
\end{proof}

\subsubsection{Bounds on amortized entanglement measures and applications}
\label{sec:amortized_entanglement}

Another task that has an adaptive nature is the task of quantum communication using free two-way classical communication.
In order to analyze such tasks, one usually considers an entanglement measure and tracks its value during the rounds of the protocol. Here, we will focus on measures of the following form: for $\alpha \in [1, \infty]$ and some convex subset $\cC(X:Y) \subseteq \Pos(X  Y)$, we can define for a bipartite state $\rho_{XY}$
\begin{align*}
E_{\alpha, \cC}(X : Y)_{\rho} = \inf_{\sigma \in \cC(X:Y)} \widetilde{D}_{\alpha}(\rho_{XY} \| \sigma_{XY}) \ .
\end{align*}
When it is clear from the context $\alpha$ and $\cC$ will be dropped from the notation. 
Note that this quantity is \change{quasiconvex} in $\rho_{XY}$. In fact using the joint quasiconvexity of $\widetilde{D}_{\alpha}$, we have for $\lambda \in [0,1]$, $\rho^0, \rho^1 \in \D(XY)$ and $\sigma^0, \sigma^1 \in \cC(X : Y)$
\begin{align*}
\widetilde{D}_{\alpha}(\lambda \rho^{0} + (1-\lambda) \rho^{1} \| \lambda \sigma^0 + (1-\lambda) \sigma^1) \leq \max\{ \widetilde{D}_{\alpha}(\rho^0 \| \sigma^0), \widetilde{D}_{\alpha}(\rho^1 \| \sigma^1)\} \ .
\end{align*}
Taking the infimum over $\sigma^0$ and $\sigma^1$, we get the quasiconvexity of $E(X : Y)_{\rho}$ in $\rho$. 
To make this a useful correlation measure, we will assume that $\cC(X:Y)$ contains all the product states $\phi_{X} \otimes \psi_{Y}$, and as we assumed convexity of $\cC(X : Y)$, it also contains the set of all separable states. 
Many studied quantum correlation measures are special cases:
\begin{itemize}
\item For the relative entropy of entanglement, $\cC$ is the set of separable states and $\alpha = 1$, but the full range $\alpha \in [1, \infty]$ has also been used, in particular for the study of adaptive protocols~\cite{PLOB17,wilde2017converse,christandl2017relative}.
\item For the Rains bound, $\cC = \mathrm{PPT'} := \{\sigma_{XY}  \in \Pos(X Y) : \| \sigma_{XY}^{\top_{Y}} \|_{1} \leq 1\}$ and $\alpha = 1$~\cite{rains2001semidefinite}, and the version with $\alpha = \infty$ has also been studied in~\cite{wang2016improved,Berta2017}. The notation $\top_{Y}$ denotes the partial transpose on system $Y$ with respect to some fixed basis.
\end{itemize}
One can then naturally define the entanglement of a quantum channel $\cN_{X \to Y}$ as
\begin{align*}
E(\cN) = \sup_{\rho \in \D(X'  X)} E(X':Y)_{(\cI_{X'} \otimes \cN_{X \to Y})(\rho_{X'X})} \ ,
\end{align*}
where the supremum \change{also} runs over arbitrary finite dimensional systems $X'$. Note that using the quasiconvexity of $E$ in the state, we may restrict $\rho_{X'X}$ to be pure. Thus, whenever the set $\cC$ is invariant under local isometries (which will be the case here), it suffices to take $X'$ to have the same dimension as $X$. The amortized version is then defined as
\begin{align*}
E^{\amor}(\cN) = \sup_{\rho_{X'XY'} \in \D(X'XY)} E(X':YY')_{(\cI_{X'Y'} \otimes \cN_{X \to Y})(\rho_{X'XY'})} - E(X'X:Y')_{\rho_{X'XY'}} \ ,
\end{align*}
where the supremum runs over arbitrary finite dimensional systems $X' Y'$. Note that if $Y'$ is trivial, we recover $E(\cN)$ but in general it is not clear how to bound the dimensions of the systems $X'$ and $Y'$. Amortized quantities allow one to place upper bounds on the rates of protocols allowing two-way communication, as shown for example~\cite{bennett2003capacities} in the context of bidirectional channel capacities and in~\cite{kaur2017amortized,Berta2017} in the context of quantum/private communication with free two-way classical communication.
For completeness, we illustrate this methodology in the following simple lemma that bounds the quantum correlations that can be obtained by a process of the form given in Figure~\ref{fig:quantum_comm_locc}.  For convenience of notation, we will be using the trivial $1$-dimensional system $Y_0$.
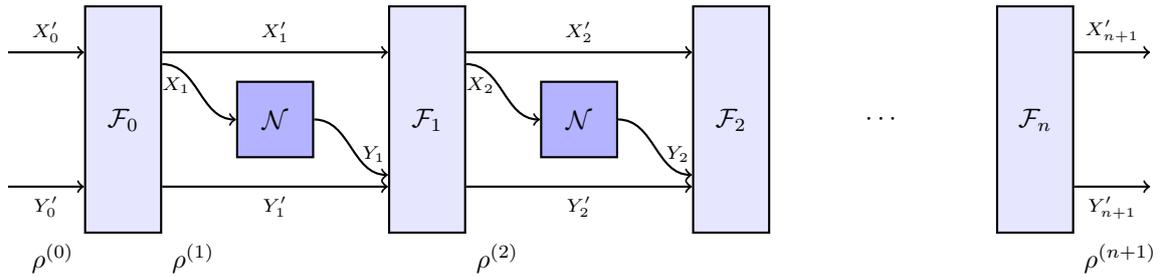
\begin{figure}[ht]
    \begin{center}
        \begin{tikzpicture}[thick]
        \tikzstyle{bipartitechannel}=[minimum width=1cm,minimum height=3cm,rectangle,draw,fill=blue!10]
        \tikzstyle{channel}=[minimum width=1cm,minimum height=1cm,rectangle,draw,fill=blue!30]        
        \tikzstyle{syslabel}=[font=\scriptsize]                    
        \draw
            (0, 0) node[bipartitechannel] (F0) {$\mathcal{F}_0$}
            ++(2, 0) node[channel] (N1) {$\mathcal{N}$}
            ++(2, 0) node[bipartitechannel] (F1) {$\mathcal{F}_1$}
            ++(2, 0) node[channel] (N2) {$\mathcal{N}$}
            ++(2, 0) node[bipartitechannel] (F2) {$\mathcal{F}_2$}
            ++(2,0) node (dotdotdot) {$\cdots$}
            ++(2, 0) node[bipartitechannel] (Fn) {$\mathcal{F}_n$}            
            ;
        
        \draw (F0.south west)   node[below left] {$\rho^{(0)}$} ;
        
        \draw (F0.south east)   node[below right] {$\rho^{(1)}$} ;
        \draw (F1.south east)   node[below right] {$\rho^{(2)}$} ;
        \draw (Fn.south east)   node[below right] {$\rho^{(n+1)}$} ;
            
	\draw[->] (F0.120) ++(-1,0) -- (F0.120) node[syslabel,midway, above]{$X'_0$};
	\draw[->] (F0.240) ++(-1,0) -- (F0.240) node[syslabel,midway, below]{$Y'_0$};
	\draw[->] (F0.60)  -- (F1.120) node[syslabel,midway, above]{$X'_1$};
	\draw[->] (F0.55) to[out=0,in=180] node[syslabel,pos=0.15,below]{$X_1$}  (N1.west) ;	
	\draw[->] (N1.east) to[out=0,in=180]  node[syslabel,pos=0.85,above]{$Y_1$} (F1.235);	
	\draw[->] (F0.300)  -- (F1.240) node[syslabel,midway, below]{$Y'_1$};

	\draw[->] (F1.60)  -- (F2.120) node[syslabel,midway, above]{$X'_2$};
	\draw[->] (F1.55) to[out=0,in=180] node[syslabel,pos=0.15,below]{$X_2$}  (N2.west) ;	
	\draw[->] (N2.east) to[out=0,in=180]  node[syslabel,pos=0.85,above]{$Y_2$} (F2.235);	
	\draw[->] (F1.300)  -- (F2.240) node[syslabel,midway, below]{$Y'_2$};	
	
	\draw[->] (Fn.60)  -- ++(1,0) node[syslabel,midway, above]{$X'_{n+1}$};
	\draw[->] (Fn.300) -- ++(1,0) node[syslabel,midway, below]{$Y'_{n+1}$};	
      \end{tikzpicture}
\end{center}
    \caption{The state $\rho^{(n+1)}_{X'_{n+1} Y'_{n+1}}$ is generated by a sequence of quantum channels as indicated in the Figure. The channels $\cF_{i}$ should be considered as free operations (e.g., modeling two-way classical communication) between Alice (top) and Bob (bottom) and $\cN$ is a quantum channel going from Alice to Bob.
    }
     \label{fig:quantum_comm_locc}
\end{figure}

\begin{lemma}
\label{lem:entanglement_tracking}
Let $\rho^{(0)}_{X'_0 Y'_0}$ be a quantum state in $\cC(X'_0 : Y'_0)$  and assume that the quantum channels $\cF_{i}$ map elements in $\cC(X'_i : Y_i Y'_i)$ to elements in $\cC(X'_{i+1} X_{i+1} : Y'_{i+1})$. Then the state $\rho^{(n+1)}_{X'_{n+1} Y'_{n+1}}$ generated as in Figure~\ref{fig:quantum_comm_locc} satisfies
\begin{align*}
E(X'_{n+1} : Y'_{n+1})_{\rho^{(n+1)}} \leq n E^{\amor}(\cN) \ .
\end{align*}
\end{lemma}
\begin{proof}
Using the definition of $E$, we can write
\begin{align*}
E(X'_{n+1} : Y'_{n+1})_{\rho^{(n+1)}} 
&= \inf_{\sigma_{X'_{n+1} Y'_{n+1}} \in \cC(X'_{n+1} : Y'_{n+1})} \widetilde{D}_{\alpha}(\rho^{(n+1)}_{X'_{n+1} Y'_{n+1}} \| \sigma_{X'_{n+1} Y'_{n+1}}) \\
&\leq \inf_{\sigma_{X'_n Y_n Y'_n} \in \cC(X'_n : Y_n Y'_n) } \widetilde{D}_{\alpha}\left(\cF_n\left( \cN_{X_n \to Y_n}(\rho^{(n)}_{X'_{n} X_n Y'_{n}} ) \right) \| \cF_{n}\left(\sigma_{X'_n Y_n Y'_n}\right) \right) \\
&\leq \inf_{\sigma_{X'_n Y_n Y'_n} \in \cC(X'_n : Y_n Y'_n) } \widetilde{D}_{\alpha}\left( \cN_{X_n \to Y_n}(\rho^{(n)}_{X'_{n} X_n Y'_{n}} ) \| \sigma_{X'_n Y_n Y'_n} \right) \\
&= E(X'_{n} : Y_n Y'_{n})_{\cN_{X_n \to Y_n}(\rho^{(n)})} \\
&\leq E^{\amor}(\cN) + E(X'_{n} X_n : Y'_{n})_{\rho^{(n)}} \ ,
\end{align*}
using the definition of the amortized quantity $E^{\amor}(\cN)$. Repeating this argument, and using the fact that $E(X'_{0} : Y'_{0})_{\rho^{(0)}} = 0$, we obtain the desired result.
\end{proof}

However, the issue with the amortized quantity $E^{\amor}(\cN)$ is that it is unclear how to compute it. Using our chain rule, one can upper bound this $E^{\amor}(\cN)$ in terms of a regularized divergence by finding channels $\cM$ having the right properties. Then one can use Theorem~\ref{thm:comp_reg_sand} to obtain computable upper bounds on the regularized divergence.
\begin{lemma}
\label{lem:bound_amortized_entanglement}
Let $\cM_{X \to Y}$ be a completely positive map satisfying the following property. For any $\rho_{X'YY'}  \in \D(X'YY')$ and any $\sigma_{X'XY'} \in \cC(X'X : Y')$, we have
\begin{align}
\label{eq:condition_channel_M}
E(X':YY')_{\rho} \leq \widetilde{D}_{\alpha}(\rho_{X'YY'} \| (\cI_{X'Y'} \otimes \cM_{X \to Y})(\sigma_{X'XY'})) \ .
\end{align} 
 Then
\begin{align*}
E^{\amor}(\cN) \leq \widetilde{D}^{\reg}_{\alpha}(\cN \| \cM).
\end{align*}
\end{lemma}
\begin{proof}
Consider $\rho_{X'XY'}$ and let $\sigma_{X'XY'} \in \cC(X'X:Y')$. Applying the chain rule, we obtain
\begin{align*}
\widetilde{D}_{\alpha}(\cN_{X \to Y}(\rho_{X'XY'}) \| \cM_{X \to Y}(\sigma_{X'XY})) \leq \widetilde{D}^{\reg}_{\alpha}(\cN \| \cM) + \widetilde{D}_{\alpha}(\rho_{X'XY'} \| \sigma_{X'XY'}) \ .
\end{align*}
Thus,
\begin{align*}
E^{\amor}(\cN) 
&= \sup_{\rho_{X'XY'} \in \D(X'XY)} E(X':YY')_{(\cI_{X'Y'} \otimes \cN_{X \to Y})(\rho_{X'XY})} - E(X'X:Y')_{\rho_{X'XY}} \\
&= \sup_{\rho_{X'XY'} \in \D(X'XY)} \sup_{\sigma_{X'XY} \in \cC(X'X:Y)} \left( E(X':YY')_{(\cI_{X'Y'} \otimes \cN_{X \to Y})(\rho_{X'XY})} -  \widetilde{D}_{\alpha}(\rho_{X'XY'} \| \sigma_{X'XY'}) \right) \\
&\leq \sup_{\rho_{X'XY'} \in \D(X'XY)} \sup_{\sigma_{X'XY} \in \cC(X'X:Y)} \left( \widetilde{D}_{\alpha}( \cN_{X \to Y}(\rho_{X'XY'}) \| \cM_{X \to Y}(\sigma_{X'XY})) -  \widetilde{D}_{\alpha}(\rho_{X'XY'} \| \sigma_{X'XY'}) \right) \\
&\leq \widetilde{D}^{\reg}_{\alpha}(\cN \| \cM) \ .
\end{align*}
\end{proof}


We could then apply this methodology to a variety of tasks. Here we consider the task of quantum communication between Alice and Bob with free classical two-way communication. For that, we will fix $\cC$ to be the set known as $\mathrm{PPT'}$~\cite{rains2001semidefinite} defined by $\cC(X : Y) = \{\sigma_{XY}  \in \Pos(X Y) : \| \sigma_{XY}^{\top_{Y}} \|_{1} \leq 1\}$ and $\alpha \in (1, \infty)$. We then have to find a set of channels $\cM$ satisfying the condition~\eqref{eq:condition_channel_M}. For that we use set of channel used in~\cite{FF19} (this choice can be traced back to~\cite{holevo2001evaluating}),
\begin{align*}
{\cV_{\Theta}} := \{ \cM \in \CP(X : Y) : \| \Theta_{Y} \circ \cM_{X \to Y} \|_{\diamond} \leq 1 \} \ ,
\end{align*}
where $\Theta_{Y}$ denotes the transpose map and the diamond norm of a linear map $\cA$ from $\Lin(X')$ to $\Lin(Y)$ is defined by $\| \cA \|_{\diamond} = \sup \{ \| (\cI_{X} \otimes \cA_{X' \to Y})(W_{XX'}) \|_1 : \| W_{XX'} \|_{1} \leq 1 \}$.
Notice that any $\cM \in \cV_{\Theta}$ satisfies the condition~\eqref{eq:condition_channel_M} as $\cM(\sigma_{X'XY'}) \in \cC(X : YY')$ for any $\sigma_{X'XY'} \in \cC(X' X : Y)$. In fact, we have for any $\sigma_{X'XY'}$ such that $\| \sigma_{X'XY'}^{\top_{Y'}} \|_{1} \leq 1$, we have
\begin{align*}
\| \cM_{X \to Y}(\sigma_{X'XY'})^{\top_{YY'}} \|_{1} 
&= \| \Theta_{Y} \circ \cM_{X \to Y}(\sigma_{X'XY'}^{\top_{Y'}}) \|_1 \\
&\leq \| \Theta_{Y} \circ \cM_{X \to Y} \|_{\diamond} \| \sigma_{X'XY'}^{\top_{Y'}} \|_{1} \\
&\leq 1 \ .
\end{align*}
\begin{proposition}
Let $\eps \in [0,1]$, $k \in \NN_+$ and consider a state $\rho^{(n+1)}_{X'_{n+1} Y'_{n+1}}$ generated as in Figure~\ref{fig:quantum_comm_locc} with the quantum channels $\cF_{i}$ that preserve the property $\mathrm{PPT'}$ (which is in particular the case for classical two-way communication and local operations). Assume that $X'_{n+1}$ and $Y'_{n+1}$ are $k$-qubit systems and that $\tr(\rho^{(n+1)}_{X'_{n+1} Y'_{n+1}} \Psi^{\otimes k}_{}) \geq 1-\eps$, with $\Psi = \proj{\Psi}$ is a maximally entangled state $\ket{\Psi} = \frac{1}{\sqrt{2}}(\ket{00} + \ket{11})$. Then,
for any $\alpha \in (1,\infty)$ and any $\cM \in \cV_{\Theta}$,
\begin{align*}
\frac{k}{n} 
&\leq \widetilde{D}^{\reg}_{\alpha}(\cN \| \cM) - \frac{\alpha}{n(\alpha - 1)} \log(1-\eps) \ .
\end{align*}
\end{proposition}
\begin{proof}
Applying Lemma~\ref{lem:entanglement_tracking} and then Lemma~\ref{lem:bound_amortized_entanglement} for the choice of $\cC$ and $\cM$ described above, we have
\begin{align*}
E(X'_{n+1} : Y'_{n+1})_{\rho^{(n+1)}} 
&\leq n E^{\amor}(\cN) \\
&\leq n \widetilde{D}_{\alpha}^{\reg}(\cN \| \cM) \ .
\end{align*}
We now want to relate the quantity $E(X'_{n+1} : Y'_{n+1})_{\rho^{(n+1)}}$ to $\eps$ and $k$. Using the data processing inequality for $\widetilde{D}_{\alpha}$ with the completely positive and trace-preserving map $\cA(W) = \tr(W \Psi^{\otimes k}) \proj{0} + (1-\tr(W \Psi^{\otimes k})) \proj{1}$, we have
\begin{align*}
E(X'_{n+1} : Y'_{n+1})_{\rho^{(n+1)}} 
&= \inf_{ \sigma_{X'_{n+1} Y'_{n+1}} \in \cC(X'_{n+1} : Y'_{n+1}) } \widetilde{D}_{\alpha}(\rho^{(n+1)}_{X'_{n+1} Y'_{n+1}} \| \sigma_{X'_{n+1} Y'_{n+1}}) \\
&\geq \inf_{ \sigma_{X'_{n+1} Y'_{n+1}} \in \cC(X'_{n+1} : Y'_{n+1}) } \delta_{\alpha}\left(\tr(\rho^{(n+1)}_{X'_{n+1} Y'_{n+1}} \Psi^{\otimes k}) \| \tr(\sigma_{X'_{n+1} Y'_{n+1}} \Psi^{\otimes k}) \right) \ ,
\end{align*}
where $\delta_{\alpha}(p \| q) = \frac{1}{\alpha - 1} \log \left( p^{\alpha} q^{1-\alpha} + (1-p)^{\alpha} (1-q)^{1-\alpha} \right)$ is the binary R\'enyi divergence. By assumption $\tr(\rho^{(n+1)}_{X'_{n+1} Y'_{n+1}} \Psi^{\otimes k}) \geq 1-\eps$ and for any state $\sigma_{X'_{n+1} Y'_{n+1}}$ that is $\mathrm{PPT}'$, we have $\tr(\sigma_{X'_{n+1} Y'_{n+1}} \Psi^{\otimes k}) \leq 2^{-k}$~\cite{rains2001semidefinite}. As a result, as $\alpha > 1$, we have
\begin{align*}
\delta_{\alpha}\left(\tr(\rho^{(n+1)}_{X'_{n+1} Y'_{n+1}} \Psi^{\otimes k}) \| \tr(\sigma_{X'_{n+1} Y'_{n+1}} \Psi^{\otimes k}) \right) 
&\geq \frac{1}{\alpha - 1} \log \left( (1-\eps)^{\alpha} 2^{k (\alpha - 1)} \right)  \\
&= \frac{\alpha}{\alpha - 1} \log(1-\eps) + k \ .
\end{align*}
Putting everything together, we obtain the desired bound.
\end{proof}

Using the fact that $\widetilde{D}_{\alpha}^{\reg}(\cN \| \cM) \leq \newD_{\alpha}(\cN \| \cM)$ and the fact that the set of channels $\cV_{\Theta}$ is representable by a semidefinite program, we obtain efficiently computable bounds $\min_{\cM \in \cV_{\Theta}} \newD_{\alpha}(\cN \| \cM)$ on the quantum capacity assisted with free $\mathrm{PPT}'$-preserving operations. As we also have $\widetilde{D}_{\alpha}^{\reg}(\cN \| \cM) \leq  \frac{1}{m} \newD_{\alpha}(\cN^{\otimes m} \| \cM^{\otimes m})$ for any $m \geq 1$, $\min_{\cM \in \cV_{\Theta}} \frac{1}{m} \newD_{\alpha}(\cN^{\otimes m} \| \cM^{\otimes m})$ is also a valid upper bound but it is not clear how to compute it efficiently when $m \geq 2$. Nonetheless, one can use the map $\cM \in \cV_{\Theta}$ that minimizes $\min_{\cM \in \cV_{\Theta}} \newD_{\alpha}(\cN \| \cM)$ and evaluate $\frac{1}{m} \newD_{\alpha}(\cN^{\otimes m} \| \cM^{\otimes m})$ for this map.
We illustrate these bounds in Figure~\ref{fig:quantum_cap} for the amplitude damping channel, where we obtain an improved bound compared to using the geometric R\'enyi divergence $\widehat{D}_{\alpha}$ in~\cite{FF19}.

\begin{figure}[ht]
  \centering
%
%
\definecolor{mycolor1}{rgb}{0.00000,0.44700,0.74100}%
\definecolor{mycolor2}{rgb}{0.85000,0.32500,0.09800}%
\definecolor{mycolor3}{rgb}{0.92900,0.69400,0.12500}%
\begin{tikzpicture}

\begin{axis}[%
width=4in,
height=2.5in,
at={(1.048in,0.726in)},
scale only axis,
xmin=0,
xmax=1,
ymin=0,
ymax=1.05,
xlabel={$\gamma$},
ylabel style={rotate=-90},
ylabel={$D$},
axis background/.style={fill=white},
legend style={legend cell align=left, align=left, draw=white!15!black}
]
\addplot [color=mycolor1, line width=1.5pt]
  table[row sep=crcr]{%
0	0.999999999818654\\
0.025	0.96718528654409\\
0.05	0.940543750734257\\
0.1	0.892500444843272\\
0.2	0.804555240572948\\
0.3	0.720122479705461\\
0.4	0.635511735376004\\
0.5	0.548461571846658\\
0.6	0.457092901698822\\
0.7	0.359366802001355\\
0.8	0.252870119674115\\
0.9	0.134475360686317\\
0.95	0.0695631387394038\\
1	0\\
};
\addlegendentry{Bound $\min_{\cM \in \cV_{\Theta}} \newD_{\alpha}(\cN_{\gamma} \| \cM)$}

\addplot [color=mycolor2, line width=1.5pt, dashed]
  table[row sep=crcr]{%
0	0.999999999818654\\
0.025	0.96718528654409\\
0.05	0.940469934357659\\
0.1	0.888566126404789\\
0.2	0.794424751480545\\
0.3	0.706361308285264\\
0.4	0.620356490546093\\
0.5	0.533858545349676\\
0.6	0.444629188190253\\
0.7	0.35022999428678\\
0.8	0.247640021262488\\
0.9	0.132803513966608\\
0.95	0.069091693177106\\
1	-1.52690683627117e-10\\
};
\addlegendentry{Bound $\frac{1}{3} \newD_{2}(\cN_{\gamma}^{\otimes 3} \| \cM^{\otimes 3})$}

\addplot [color=mycolor3, line width=1.5pt, dotted]
  table[row sep=crcr]{%
0	1.0000122376519\\
0.025	0.981864556738505\\
0.05	0.963487028379919\\
0.1	0.926003073430793\\
0.2	0.848000492348132\\
0.3	0.765541627600207\\
0.4	0.678077466946298\\
0.5	0.584967868575473\\
0.6	0.48543167442716\\
0.7	0.378517121200577\\
0.8	0.263035706749808\\
0.9	0.137512932339375\\
0.95	0.0703955823322711\\
1	0\\
};
\addlegendentry{Bound $\min_{\cM \in \cV_{\Theta}} \widehat{D}_{\alpha}(\cN_{\gamma} \| \cM)$~\cite{FF19}}

\end{axis}
\end{tikzpicture}%
  \caption{Upper bounds on the quantum capacity with free $\mathrm{PPT}'$ preserving operations for the amplitude damping channel defined by $\cN_{\gamma}(\rho) = (\proj{0} + \sqrt{1-\gamma} \proj{1}) \rho (\proj{0} + \sqrt{1-\gamma} \proj{1}) + \gamma \ketbra{0}{1} \rho \ketbra{1}{0}$. The solid plot shows $\min \left\{ \min_{\cM \in \cV_{\Theta}} \newD_{\alpha}(\cN_{\gamma} \| \cM) : \alpha \in \{1.1,1.2, \dots, 2\}\right\}$. The dashed plot shows $\frac{1}{3} \newD_{2}(\cN_{\gamma}^{\otimes 3} \| \cM^{\otimes m})$ where $\cM = \argmin_{\cM \in \cV_{\Theta}} \newD_2(\cN_{\gamma} \| \cM)$ and we observe a slightly improved bound compared to the solid plot. The dotted plot shows the bound obtained using $\widehat{D}_{\alpha}$ from~\cite{FF19}, which happens to match with the bound based on $D_{\max}$ for the amplitude damping channel as shown in~\cite{FF19}.}
  \label{fig:quantum_cap}
\end{figure}

\section{Discussion}
\label{sec:discussion}

We have presented a family of quantum $\alpha$-R\'enyi divergences for $\alpha > 1$ based on the geometric mean. The framework is in fact more general and allows us to define quantum divergences in a similar way using a Kubo-Ando mean for any operator monotone function $g : [0,\infty) \to [0, \infty)$. As we mostly used generic properties of operator means to establish properties of $\newD_{\alpha}$, we expect analogous properties for more general functions $g$ to hold. \change{For example, for any convex function $f$, the $f$-divergence between distributions $P$ and $Q$ is defined as $Q_{f}(P \| Q) = \sum_{x} f(\frac{P(x)}{Q(x)}) Q(x)$. When $f(t) = t^{\alpha}$ with $\alpha > 1$, we obtain (after applying $\frac{1}{\alpha - 1} \log$) the $\alpha$-R\'enyi divergence. Several quantum $f$-divergences have been proposed, see e.g.,~\cite{hiai2017different}. In the special case where $f$ is bijective and its inverse $f^{-1}$ is operator monotone, then using the Kubo-Ando mean associated with $g = f^{-1}$, we would obtain a quantum version of the $f$-divergence. Here, we focused on the case $f(t) = t^{\alpha}$ and correspondingly $g(t) = t^{1/\alpha}$, but it would be interesting to explore other choices of $g$ and potential applications.}
In a different direction, a variant of $\newD_{\alpha}$ is defined in~\cite{BFF20}, using the $\frac{1}{2}$-geometric mean but one takes the geometric mean $k$ times iteratively with different variables. More generally, we hope that our work encourages the study of further quantum divergences that are defined via convex optimization programs.

We leave multiple open questions. A specific question is whether $\lim_{\alpha \to 1} \newD_{\alpha}(\rho \| \sigma)$ is equal to the Belavkin-Staszewski divergence $\widehat{D}(\rho \| \sigma)$~\cite{BS82}? Numerical examples suggest that this should be the case. 
Another question is whether it is possible to define $\newD_{\alpha}$ when $\alpha < 1$ with similar properties? The natural extension would be to define $\newD_{\alpha}(\rho \| \sigma) = \frac{1}{\alpha-1} \log \newQ_{\alpha}(\rho \| \sigma)$ with $\newQ_{\alpha}(\rho \| \sigma) = \max\{ \tr(A) : \rho \geq \sigma \#_{1/\alpha} A \}$. But with this definition, it is simple to check using the operator monotonicity of $t \mapsto t^{\alpha}$ for $\alpha \in [0,1]$ that $\newD_{\alpha}(\rho \| \sigma) = \widehat{D}_{\alpha}(\rho \| \sigma)$ which means that we cannot have the property~\eqref{eq:reg_im_sand} for example. This argument does not go through when $\alpha > 1$ as $t \mapsto t^{\alpha}$ is not operator monotone in this regime. It would also be interesting to generalize the divergences introduced here to infinite-dimensional spaces or even to von Neumann algebras. Another important question that is left open is whether $\widetilde{D}_{\alpha}^{\reg}$ converges to $D^{\reg}$ when $\alpha \to 1$.

\section*{Acknowledgements}

We would like to thank Kun Fang, Renato Renner and David Sutter for discussions about the chain rule for quantum divergences, Peter Brown for discussions about quantum divergences and for comments on the manuscript, Alexander M\"uller-Hermes for comments on the monotonicity of $\newD_{\alpha}$ under positive maps, Mario Berta for discussions about quantum channel discrimination and Mark Wilde for comments on a previous draft. We are also very grateful to the anonymous reviewers for their detailed feedback that significantly improved the paper. This research is supported by the French National Research Agency via Project No. ANR-18-CE47-0011 (ACOM). 

\bibliographystyle{plainnat}
\bibliography{big}

\appendix

\section{Various results}

The following standard lemma about permutation invariant operators was used for the proof of Lemma~\ref{lem:reg_channel_div}. 
\begin{lemma}
\label{lem:perm_inv_spec}
Let $X$ be a permutation-invariant operator on $(\CC^d)^{\otimes n}$, i.e., $[P({\pi}), X] = 0$ for any permutation $\pi \in \mathfrak{S}_n$ and $P({\pi}) \ket{\psi_1} \otimes \cdots \otimes \ket{\psi_n} = \ket{\psi_{\pi^{-1}(1)}} \otimes \cdots \otimes \ket{\psi_{\pi^{-1}(1)}}$. Then 
\begin{align*}
|\mathrm{spec}(X)| \leq (n+1)^d (n+d)^{d^2} \ .
\end{align*}
\end{lemma}
\begin{proof}
$P$ defines a representation of the symmetric group $\mathfrak{S}_n$ on $(\CC^d)^{\otimes n}$ and its decomposition into irreducible representations is well-known, see e.g.,~\cite[Section 5.3]{Harrow05}. In fact, its irreducible representations are labelled by the set $\cI_{n,d}$ of Young diagrams of size $n$ with at most $d$ rows. For $\lambda \in \cI_{n,d}$, we denote by $p_{\lambda}$ the corresponding irreducible representation acting on the space $V_{\lambda}$. Each $p_{\lambda}$ appears in general multiple times in $P$ and this is taken into account by introducing the multiplicity space $U_{\lambda}$ (which happens to correspond to an irreducible representation of the unitary group but we will not use this here). Summarizing, the operator $P(\pi)$ can in the Schur basis be written as 
\begin{align*}
P(\pi)  = \sum_{\lambda \in \cI_{n,d}} \proj{\lambda} \otimes I_{U_{\lambda}} \otimes p_{\lambda}(\pi) \ .
\end{align*}
We now express the operator $X$ in the Schur basis
\begin{align*}
X &= \sum_{\substack{\lambda, \lambda' \in \cI_{n,d} \\ i \in [m(\lambda)] ,i' \in [m(\lambda')] }} \ketbra{\lambda}{\lambda'} \otimes \ketbra{u_{\lambda,i}}{u_{\lambda',i'}} \otimes X_{(\lambda, i), (\lambda', i')} \ ,
\end{align*}
where we have introduced orthonormal bases $\{u_{\lambda, i}\}_{i \in m(\lambda)}$ of the spaces $U_{\lambda}$ ($m(\lambda)$ is the dimension of $U_{\lambda}$) and $X_{(\lambda, i), (\lambda', i')}$ can be seen as an operator from $V_{\lambda'}$ to $V_{\lambda}$. We can now write the products $P(\pi) X$ and $X P(\pi)$ as 
\begin{align*}
P(\pi) X  &= \sum_{\substack{\lambda, \lambda' \in \cI_{n,d} \\ i \in [m(\lambda)] ,i' \in [m(\lambda')] }} \ketbra{\lambda}{\lambda'} \otimes \ketbra{u_{\lambda,i}}{u_{\lambda',i'}} \otimes p_{\lambda}(\pi) X_{(\lambda, i), (\lambda', i')}  \\
X P(\pi)  &= \sum_{\substack{\lambda, \lambda' \in \cI_{n,d} \\ i \in [m(\lambda)] ,i' \in [m(\lambda')] }} \ketbra{\lambda}{\lambda'} \otimes \ketbra{u_{\lambda,i}}{u_{\lambda',i'}} \otimes  X_{(\lambda, i), (\lambda', i')} p_{\lambda'}(\pi)  \ .
\end{align*}
Applying Schur's lemma, we get that $X_{(\lambda, i), (\lambda', i')} = 0$ if $\lambda \neq \lambda'$ and $X_{(\lambda, i), (\lambda, i')} = x_{\lambda,i,i'} I_{V_{\lambda}}$ for some scalar $x_{\lambda, i,i'}$. Defining the operator $X_{U_{\lambda}} = \sum_{i,i' \in m(\lambda)} x_{\lambda, i,i'} \ketbra{u_{\lambda,i}}{u_{\lambda,i'}}$, we can write $X$ as
\begin{align*}
X = \sum_{\lambda \in \cI_{n,d}} \proj{\lambda} \otimes X_{U_{\lambda}} \otimes \id_{V_{\lambda}} \ .
\end{align*} 
As a result,  $|\mathrm{spec}(X)| \leq |\cI_{n,d}| \max_{\lambda} m(\lambda)$. But it is well-known that $|\cI_{n,d}| \leq (n+1)^d$ and $\max_{\lambda} m(\lambda) \leq (n+d)^{d^2}$  (see e.g.,~\cite[Section 6.2]{Harrow05}). This concludes the proof of the claim.
\end{proof}

We also need the following concavity and continuity statement. 
\begin{lemma}
\label{lem:concave_omega}
Let $\cN, \cM$ be completely positive maps from $\Lin(X)$ to $\Lin(Y)$ and $\alpha > 1$. The function 
\begin{align*}
\omega \in \D(X) \mapsto \widetilde{D}_{\alpha}( \omega^{\frac{1}{2}} J^{\cN} \omega^{\frac{1}{2}} \| \omega^{\frac{1}{2}} J^{\cM} \omega^{\frac{1}{2}}) 
\end{align*}
is concave and thus continuous on $\{\omega \in \D(X) : \omega > 0\}$. In addition, it is lower semicontinuous on $\D(X)$ and as a result, we have
\begin{align}
\sup_{\omega \in \D(X)} \widetilde{D}_{\alpha}( \omega^{\frac{1}{2}} J^{\cN} \omega^{\frac{1}{2}} \| \omega^{\frac{1}{2}} J^{\cM} \omega^{\frac{1}{2}}) 
&= \sup_{\substack{\omega \in \D(X) \\ \omega > 0}} \widetilde{D}_{\alpha}( \omega^{\frac{1}{2}} J^{\cN} \omega^{\frac{1}{2}} \| \omega^{\frac{1}{2}} J^{\cM} \omega^{\frac{1}{2}}) \ .
\label{eq:sup_channel_div_strict}
\end{align}
\end{lemma}
\begin{proof}
The concavity is very similar to the argument in~\cite{WFT17} for $\alpha = 1$.
Let $\omega_0, \omega_1 \in \D(X), \lambda \in [0,1]$ and $\omega = (1-\lambda) \omega_0 + \lambda \omega_1$. Note that both $\omega^{\frac{1}{2}} \ket{\Phi}_{XX'}$ and $\sqrt{(1-\lambda)} \ket{0} \otimes \omega_{0}^{\frac{1}{2}} \ket{\Phi}_{XX'} + \sqrt{\lambda} \ket{1} \otimes \omega_{1}^{\frac{1}{2}} \ket{\Phi}_{XX'}$ are purifications of the state $\omega$. By the isometric equivalence between purifications, we have
\begin{align*}
&\widetilde{D}_{\alpha}( \omega^{\frac{1}{2}} J^{\cN} \omega^{\frac{1}{2}} \| \omega^{\frac{1}{2}} J^{\cM} \omega^{\frac{1}{2}}) \\
&= \widetilde{D}_{\alpha}\Big( (1-\lambda) \proj{0} \otimes \omega_0^{\frac{1}{2}} J^{\cN} \omega_0^{\frac{1}{2}} + \lambda \proj{1} \otimes \omega_1^{\frac{1}{2}} J^{\cN} \omega_1^{\frac{1}{2}} + \sqrt{\lambda(1-\lambda)} \ketbra{0}{1} \otimes \omega_0^{\frac{1}{2}} J^{\cN} \omega_1^{\frac{1}{2}} + \sqrt{\lambda(1-\lambda)} \ketbra{1}{0} \otimes \omega_1^{\frac{1}{2}} J^{\cN} \omega_0^{\frac{1}{2}}  \\
&\qquad \Big\| (1-\lambda) \proj{0} \otimes \omega_0^{\frac{1}{2}} J^{\cM} \omega_0^{\frac{1}{2}} + \lambda \proj{1} \otimes \omega_1^{\frac{1}{2}} J^{\cM} \omega_1^{\frac{1}{2}} + \sqrt{\lambda(1-\lambda)} \ketbra{0}{1} \otimes \omega_0^{\frac{1}{2}} J^{\cM} \omega_1^{\frac{1}{2}} + \sqrt{\lambda(1-\lambda)} \ketbra{1}{0} \otimes \omega_1^{\frac{1}{2}} J^{\cM} \omega_0^{\frac{1}{2}} \Big) \\
&\geq \widetilde{D}_{\alpha}\left( (1-\lambda) \proj{0} \otimes \omega_0^{\frac{1}{2}} J^{\cN} \omega_0^{\frac{1}{2}} + \lambda \proj{1} \otimes \omega_1^{\frac{1}{2}} J^{\cN} \omega_1^{\frac{1}{2}} \|  (1-\lambda) \proj{0} \otimes \omega_0^{\frac{1}{2}} J^{\cM} \omega_0^{\frac{1}{2}} + \lambda \proj{1} \otimes \omega_1^{\frac{1}{2}} J^{\cM} \omega_1^{\frac{1}{2}} \right) \ ,
\end{align*}
where we used the data-processing inequality in the last line.
Now writing $\widetilde{D}_{\alpha} = \frac{1}{\alpha - 1} \log \widetilde{Q}_{\alpha}$, and using the definition of $\widetilde{Q}_{\alpha}$, we have
\begin{align*}
&\widetilde{Q}_{\alpha}\left( (1-\lambda) \proj{0} \otimes \omega_0^{\frac{1}{2}} J^{\cN} \omega_0^{\frac{1}{2}} + \lambda \proj{1} \otimes \omega_1^{\frac{1}{2}} J^{\cN} \omega_1^{\frac{1}{2}} \| (1-\lambda) \proj{0} \otimes \omega_0^{\frac{1}{2}} J^{\cM} \omega_0^{\frac{1}{2}} + \lambda \proj{1} \otimes \omega_1^{\frac{1}{2}} J^{\cM} \omega_1^{\frac{1}{2}} \right) \\
&= (1-\lambda) \widetilde{Q}_{\alpha}(\omega_0^{\frac{1}{2}} J^{\cN} \omega_0^{\frac{1}{2}} \| \omega_0^{\frac{1}{2}} J^{\cM} \omega_0^{\frac{1}{2}}) + \lambda \widetilde{Q}_{\alpha}(\omega_1^{\frac{1}{2}} J^{\cN} \omega_1^{\frac{1}{2}} \| \omega_1^{\frac{1}{2}} J^{\cM} \omega_1^{\frac{1}{2}}) \ .
\end{align*}
Using the concavity of the logarithm, we finally obtain
\begin{align*}
\widetilde{D}_{\alpha}( \omega^{\frac{1}{2}} J^{\cN} \omega^{\frac{1}{2}} \| \omega^{\frac{1}{2}} J^{\cM} \omega^{\frac{1}{2}}) 
&\geq (1-\lambda) \widetilde{D}_{\alpha}(\omega_0^{\frac{1}{2}} J^{\cN} \omega_0^{\frac{1}{2}} \| \omega_0^{\frac{1}{2}} J^{\cM} \omega_0^{\frac{1}{2}}) + \lambda \widetilde{D}_{\alpha}(\omega_1^{\frac{1}{2}} J^{\cN} \omega_1^{\frac{1}{2}} \| \omega_1^{\frac{1}{2}} J^{\cM} \omega_1^{\frac{1}{2}}) .
\end{align*}
And it is well-known that concavity implies continuity on the relative interior~\cite[Theorem 10.1]{rockafellar1970convex}.

The lower semicontinuity follows from the continuity of $\omega \mapsto (\omega^{\frac{1}{2}} J^{\cN} \omega^{\frac{1}{2}}, \omega^{\frac{1}{2}} J^{\cM} \omega^{\frac{1}{2}})$ and the lower semicontinuity of $\widetilde{D}_{\alpha}$ from Lemma~\ref{lem:lower_semicont}.

To show~\eqref{eq:sup_channel_div_strict}, for any $\omega \in \D(X)$ let $\omega_n = (1-\frac{1}{n}) \omega + \frac{1}{n} \frac{I}{\dim X}$. Then $\omega_n > 0$ for all $n$ and by lower semicontinuity,
\begin{align*}
\widetilde{D}_{\alpha}( \omega^{\frac{1}{2}} J^{\cN} \omega^{\frac{1}{2}} \| \omega^{\frac{1}{2}} J^{\cM} \omega^{\frac{1}{2}}) 
&\leq \liminf_{n \to \infty} \widetilde{D}_{\alpha}( \omega_n^{\frac{1}{2}} J^{\cN} \omega_n^{\frac{1}{2}} \| \omega_n^{\frac{1}{2}} J^{\cM} \omega_n^{\frac{1}{2}}) \\
&\leq \sup_{\substack{\theta \in \D(X) \\ \theta > 0}} \widetilde{D}_{\alpha}( \theta^{\frac{1}{2}} J^{\cN} \theta^{\frac{1}{2}} \| \theta^{\frac{1}{2}} J^{\cM} \theta^{\frac{1}{2}}) \ .
\end{align*}
Taking the supremum over $\omega \in \D(X)$ gives the desired equality.
\end{proof}

The following lower semicontinuity statement is standard (see e.g,~\cite[Corollary 3.27]{mosonyi2017strong} and ~\cite[Proposition 3.10]{jenvcova2018renyi} for the general von Neumann algebra setting). We include a proof for completeness.
\begin{lemma}
\label{lem:lower_semicont}
For any $\alpha \in (1,\infty)$, the function 
\begin{align*}
\widetilde{D}_{\alpha} \: : \: &\Pos(\cH) \times \Pos(\cH) \to \RR \cup \{\infty\} \\
&(\rho, \sigma) \mapsto \widetilde{D}_{\alpha}(\rho \| \sigma)
\end{align*}
is lower semicontinuous. 
\end{lemma}
\begin{proof}
Let $(\rho, \sigma)$ be such that $\rho \ll \sigma$. Our objective is to show that for any sequence $(\rho_n, \sigma_n)$ converging to $(\rho, \sigma)$, we have $\liminf_{n \to \infty} \widetilde{D}_{\alpha}(\rho_n \| \sigma_n) \geq \widetilde{D}_{\alpha}(\rho \| \sigma)$. We first observe that if we restrict ourselves to the set $\Pos(\supp(\sigma)) \times \Pos(\supp(\sigma))$, then $\widetilde{D}_{\alpha}$ is continuous at $(\rho, \sigma)$ as a composition of continuous functions. Now let $P$ be the projector onto $\supp(\sigma)$. Using the data-processing inequality and the fact that $P (I-P) = 0$, we have
\begin{align*}
\widetilde{Q}_{\alpha}(\rho_n \| \sigma_n) 
&\geq \widetilde{Q}_{\alpha}\left(P\rho_nP + (I - P) \rho_n (I-P) \| P\sigma_nP + (I-P) \sigma_n (I-P) \right) \\
&= \widetilde{Q}_{\alpha}(P\rho_nP \| P\sigma_nP ) + \widetilde{Q}_{\alpha}((I-P)\rho_n(I-P) \| (I-P)\sigma_n(I-P) ) \\
&\geq \widetilde{Q}_{\alpha}(P\rho_nP \| P\sigma_nP ) \ .
\end{align*}
Now the sequence $(P \rho_n P, P \sigma_n P)$ is in $\Pos(\supp(\sigma)) \times \Pos(\supp(\sigma))$ and converges to $(\rho, \sigma)$. Using continuity of the function restricted to this set we obtain the desired result. We remark that $\widetilde{D}_{\alpha}$ is not continuous in general even in the classical case: consider for example $\rho_n = \proj{0} + n^{-1} \proj{1}$ and $\sigma_n = \proj{0} + n^{-\frac{\alpha}{\alpha-1}} \proj{1}$ with $\rho = \sigma = \proj{0}$ then $\widetilde{D}_{\alpha}(\rho \| \sigma) = 0$ but $\widetilde{D}_{\alpha}(\rho_n \| \sigma_n) = \frac{1}{\alpha - 1} \log\left(1 + n^{-\alpha} n^{+\alpha}\right) = \frac{1}{\alpha - 1}$ is bounded away from $0$.

Now assume that we do not have $\rho \ll \sigma$. In this case, our objective is to show that for any sequence $(\rho_n, \sigma_n)$ converging to $(\rho, \sigma)$, we have $\widetilde{D}_{\alpha}(\rho_n \| \sigma_n) \to \infty$ as $n \to \infty$. Note that $\supp(\rho) \not\subseteq \supp (\sigma)$ implies that $\supp (\sigma)^{\perp} \not\subseteq \supp(\rho)^{\perp}$. Let $\ket{v_1} \in \supp(\sigma)^{\perp}$ but not in $\supp(\rho)^{\perp}$. Then complete it $\{\ket{v_1}, \dots, \ket{v_d}\}$ into an orthonormal basis of $\cH$ and define the completely positive and trace-preserving map $\cM$ by $\cM(W) = \sum_{i} \proj{v_i} W \proj{v_i}$. By the data-processing inequality, we have for any $n$, 
\begin{align*}
\widetilde{D}_{\alpha}(\rho_n \| \sigma_n) 
&\geq \widetilde{D}_{\alpha}(\cM(\rho_n) \| \cM(\sigma_n)) \\
&\geq \frac{1}{\alpha - 1} \log (\bra{v_1} \rho_n \ket{v_1})^{\alpha} (\bra{v_1} \sigma_n \ket{v_1})^{1-\alpha} \ .
\end{align*}
But then $\lim_{n \to \infty} \bra{v_1} \rho_n \ket{v_1} = \bra{v_1} \rho \ket{v_1} > 0$ (as $\ket{v_1} \not\in \supp(\rho)^{\perp}$) and $\lim_{n \to \infty} \bra{v_1} \sigma_n \ket{v_1} = \bra{v_1} \sigma \ket{v_1} = 0$, which leads to the desired result.
\end{proof}

\end{document}